\documentclass[reqno,onecolumn,letterpaper,11pt]{amsart}
\usepackage{graphics,url,color,setspace,fullpage}
\usepackage{epsfig,psfrag}
\usepackage{epstopdf}
\usepackage{latexsym}
\usepackage{amsmath}
\usepackage{amssymb}
\usepackage{amstext}
\usepackage[latin1]{inputenc}
\usepackage{booktabs}
\usepackage{subfig}
\usepackage[titletoc]{appendix}

\usepackage{amsthm}
\usepackage{mathtools}
\usepackage{enumitem}
\usepackage{algorithmic}
\usepackage{algorithm}
\usepackage[pagebackref=false]{hyperref}
\usepackage[figure,table]{hypcap} 
\usepackage{xcolor}
\hypersetup{
   colorlinks=true,
   bookmarksnumbered=true,
   pdfstartview={FitH},
   pdfpagemode={UseOutlines},
   breaklinks=true,
   pdftitle={TSP Tours in Cubic Graphs: Beyond 4/3},
   pdfauthor={Correa, Larr\'e, Soto},
   citecolor=blue,bookmarksopen=false,draft=false,plainpages=false
}

\makeatletter
\providecommand*{\toclevel@algorithm}{0}
\makeatother

\newtheorem{thm}{Theorem}

\newtheorem{lem}{Lemma}
\newtheorem{prop}{Proposition}
\theoremstyle{remark}
\newtheorem{rem}{Remark}

\newcommand{\C}{\mathcal{C}}
\newcommand{\F}{\mathcal{F}}

\newcommand{\N}{\mathbb{N}}
\newcommand{\UU}{\textrm{U}}
\newcommand{\calD}{\mathcal{D}}
\newcommand{\OPT}{\mathrm{OPT}}
\newcommand{\SUB}{\mathrm{SUB}}

\title{TSP Tours in Cubic Graphs: Beyond $4/3$}
\author{Jos\'e R.~Correa}
\author{Omar Larr\'e}
\author{Jos\'e A.~Soto}
\dedicatory{Universidad de Chile}

\thanks{Partially supported by N{\'u}cleo Milenio Informaci{\'o}n y Coordinaci{\'o}n en Redes ICM/FIC P10-024F. An extended abstract of a preliminary version of this work appeared in ESA 2012.}

\keywords{TSP, Cubic Graphs, Integrality Gap}
\begin{document}

\begin{abstract}
After a sequence of improvements Boyd, Sitters, van der Ster, and
Stougie proved that any 2-connected graph whose $n$ vertices have
degree 3, i.e., a cubic 2-connected graph, has a Hamiltonian tour of
length at most $(4/3)n$, establishing in particular that the
integrality gap of the subtour LP is at most 4/3 for cubic 2-connected
graphs and matching the conjectured value of the famous 4/3
conjecture. In this paper we improve upon this result by designing an
algorithm that finds a tour of length $(4/3-1/61236)n$, implying that
cubic 2-connected graphs are among the few interesting classes of
graphs for which the integrality gap of the subtour LP is strictly
less than 4/3. With the previous result, and by considering an even
smaller $\epsilon$, we show that the integrality gap of the TSP
relaxation is at most $4/3- \epsilon$ even if the graph is not
2-connected (i.e.~for cubic connected graphs), implying that the
approximability threshold of the TSP in cubic graphs is strictly below
$4/3$. Finally, using similar techniques we show, as an additional
result, that every Barnette graph admits a tour of length at most
$(4/3 - 1/18)n$.
\end{abstract}

\maketitle

\section{Introduction}\label{section:intro}

The traveling salesman problem (TSP) in metric graphs is a landmark problem in combinatorial optimization and theoretical computer science. Given a graph in which  edge-distances form a metric the goal is to find a tour of minimum length visiting each vertex at least once. In particular, understanding the approximability of the TSP has attracted much attention since Christofides \cite{C76} designed a $3/2$-approximation algorithm for
the problem. Despite the great efforts, Christofides' algorithm continues to be
the current champion, while the best known lower bound, recently obtained by Lampis \cite{L12b} states that the problem is NP-hard to approximate within a factor 185/184, which improved upon the work of Papadimitriou
and Vempala \cite{PV06}. Very recently Karpinski and Schmied \cite{KS13} obtain explicit inapproximability bounds for the cases of cubic and sub cubic graphs. A key lower bound to study the approximability of
the problem is the so-called subtour elimination linear program which has long
been known to have an integrality gap of at most $3/2$
\cite{W80}. Although no progress has been made in decreasing the integrality gap
of the subtour elimination LP, its value is conjectured to be $4/3$ (see e.g.
Goemans \cite{G95}).

\subsection{Recent related work}
There have been several improvements for important special cases of the metric TSP in the last couple of years. Oveis Gharan, Saberi, and Singh
\cite{OSS11} design a $(3/2-\epsilon)$-approximation algorithm for the case of graph metrics, while M\"omke and Svensson \cite{MS11} improve that to
$1.461$, using a different approach. Mucha \cite{M12} then showed that the approximation guarantee of the M\"omke and Svensson algorithm is $13/9$.
Finally, still in the shortest path metric case, Seb\"o and Vygen \cite{SV12} find an algorithm with a guarantee of $7/4$. These results in particular
show that the integrality gap of the subtour LP is below $3/2$ in case the metric comes from a graph. Another notable direction of recent improvements
concern the $s-t$ path version of the TSP on arbitrary metrics, where the natural extension of Christofides' heuristic guarantees a solution within a
factor of $5/3$ of optimum. An, Shmoys, and Kleinberg \cite{ASK12}, find a $(1+\sqrt{5})/2$-approximation algorithm this version of
the TSP, while Seb\"o \cite{S12} further improved this result obtaining a $8/5$-approximation algorithm.

This renewed interest in designing algorithms for the TSP in graph metrics has also reached the case when we further restrict to graph metrics induced by special classes of graphs. Gamarnik, Lewenstein, and Sviridenko \cite{GLS05} show that in a 3-connected cubic graph on $n$ vertices there is always a TSP tour --visiting each vertex at least once-- of length at most $(3/2-5/389)n$, improving upon Christofides' algorithm for this graph class. A few years later, Aggarwal, Garg, and Gupta \cite{AGG11} improved the result obtaining a bound of $(4/3)n$ while Boyd et al. \cite{BSSS11} showed that the $(4/3)n$ bound holds even if only 2-connectivity assumed. Finally, M\"omke and Svensson \cite{MS11} approach allows to prove that the $(4/3)n$ bound holds for subcubic 2-connected graphs. Interestingly, the latter bound happens to be tight and thus it may be tempting to conjecture that there are cubic graphs on $n$ vertices for which no TSP tour shorter than $(4/3)n$ exists. In this paper we show that this is not the
case. Namely, we prove that any 2-connected cubic graph on $n$ vertices has a TSP tour of length $(4/3-\epsilon)n$, for $\epsilon=1/61236>0.000016$. We further refine this result and establish that for cubic graphs, not necessarily 2-connected, there exists a  $4/3-\epsilon'$ approximation algorithm for the TSP, where $\epsilon'=\epsilon/(3 + 3\epsilon)$.

On the other hand, Qian, Schalekamp, Williamson, and van Zuylen \cite{QSWZ12} show that the integrality gap of the subtour LP is strictly less than 4/3 for metrics where all the distances are either 1 or 2. Their result, based on the work of Schalekamp, Williamson, and van Zuylen \cite{SWZ12}, constitutes the first relevant special case of the TSP for which the integrality gap of the subtour LP is strictly less than 4/3. Our result implies
in particular that the integrality gap of the subtour LP is also strictly less than 4/3 in connected cubic graphs.

From a graph theoretic viewpoint, our result can also be viewed as a step towards resolving Barnette's \cite{B69} conjecture, stating that every bipartite, planar, 3-connected, cubic graph is Hamiltonian (a similar conjecture was first formulated by Tait \cite{T84}, then refuted by Tutte \cite{T46}, then reformulated by Tutte  and refuted by Horton, see e.g., \cite{BM76}, and finally reformulated by Barnette more than 40 years ago). It
is worth mentioning that for Barnette's graphs (i.e., those with the previous properties) on $n$ vertices it is straightforward to construct TSP tours of length at most $(4/3)n$, however no better bound was known. Of course, our result improves upon this, and furthermore, in this class of graphs our bound improves to $(4/3-1/18)n< 1.28n$.

\subsection{Our approach}
An \emph{Eulerian subgraph cover} (or simply a \emph{cover}) is a collection $\Gamma =\{\gamma_1,\dots, \gamma_j\}$ of connected multi-subgraphs of $G$, called \emph{components}, satisfying that (i) every vertex of $G$ is covered by exactly one component, (ii) each component is Eulerian and (iii) no edge appears more than twice in the same component. Every cover $\Gamma$ can be transformed into a TSP tour $T(\Gamma)$ of the entire graph by contracting each component, adding a doubled spanning tree in the contracted graph (which is connected) and then uncontracting the components. Boyd et al~\cite{BSSS11} defined the \emph{contribution} of a vertex $v$ in a cover $\Gamma$, as $z_\Gamma(v) = \frac{(\ell + 2)}{h}$, where $\ell$ and $h$ are the number of edges and vertices (respectively) of the component of $\Gamma$ in which $v$ lies. The extra 2 in the numerator is added for the cost of the double edge used to connect the component to the other in the spanning tree mentioned above, so that $\sum_{
v \in V} z_\Gamma(v)$ equals the number of edges in the final TSP tour $T(\Gamma)$ plus 2. Let $\calD= \{(\Gamma_i, \lambda_i)\}_{i=1}^k$, be a distribution over covers of a graph. This is, each $\Gamma_i$ is a cover of $G$ and each $\lambda_i$ is a positive number so that $\sum_{i=1}^k \lambda_i = 1$. The \emph{average contribution} of a vertex $v$ with respect to distribution $\calD$ is defined as $z_\calD(v) = \sum_{i=1}^k \lambda_i z_{\Gamma_i}(v)$.

As the starting point of our article, we study short TSP tours on Barnette graphs (bipartite, planar, cubic and 3-connected). This type of graphs is very special as it is possible to partition its set of faces into three cycle covers i.e., Eulerian subgraph covers in which every component is a cycle.  By a simple counting argument, one of these cycle covers is composed of less than $n/6$ cycles. By using the transformation described on the previous paragraph, we can obtain a TSP tour of length at most $4n/3$. To get shorter TSP tours we describe a simple procedure that perform local operations to reduce the number of cycles of the three initial cycle covers. In a nutshell, the local operations consists in the iterative alternation of faces. That is, we replace the current cycle cover $\C$ by the symmetric difference between $\C$ and the edges of a face, provided that the resulting graph is a cycle cover with fewer cycles. By repeating this process on the three initial cycle covers, it is possible to reach a
cover with fewer than $5n/36$ cycles, which, in turn, implies the existence of a TSP tour of length at most $(4/3-1/18)n$. This result is described in Section~\ref{section:Barnette}.

The idea of applying local operations to decrease the number of components of an Eulerian subgraph cover can still be applied on more general graph classes. Given a 2-connected cubic graph $G$, Boyd et al.~\cite{BSSS11} find a TSP tour $T$ of $G$ with at most $\frac43 |V(G)|-2$ edges. Their approach has two phases. In the first phase, they transform $G$ into a simpler cubic 2-connected graph $H$ not containing certain ill-behaved structures (called $p$-rainbows, for $p\geq 1$). In the second phase, they use a linear programming approach to find a polynomial collection of perfect matchings
for $H$ such that a convex combination of them gives every edge a weight of $1/3$. Their complements induce a distribution over cycle covers of $H$. By performing certain local operations on each cover, they get a distribution of Eulerian subgraph covers having average vertex contribution bounded above by 4/3. They use this to find a TSP tour for $H$ with at most $\frac43 |V(H)| - 2$ edges, which can be easily transformed into a TSP tour of $G$ having the desired guarantee. The local operations used by Boyd et al.~\cite{BSSS11} consists of iterative alternation of 4-cycles and a special type of alternation of 5-cycles (in which some edges get doubled).

Our main result is an improvement on Boyd et al.'s technique that allows us to show that every 2-connected cubic graph $G$ admits a TSP tour with at most $(4/3 - \epsilon)|V(G)| -2$ edges. The first difference between the approaches, described in Section \ref{section:simplification}, is that our simplification phase is more aggressive. Specifically, we set up a framework to eliminate large families of structures that we use to get rid of all chorded 6-cycles. This clean-up step can very likely be extended to larger families and may ultimately lead to improved results when combined with an appropriate second phase.
The second difference, described in Section \ref{section:main}, is that we extend the set of local operations of the second phase by allowing the alternation of 6-cycles of a Eulerian subgraph cover. Again, it is likely that one can further exploit this idea to get improved guarantees. Mixing these new ideas appropriately requires significant work but ultimately leads us to find a distribution $\calD$ of covers of the simplified graph $H$ for which $\sum_{v \in V(H)} z_\calD(v) \leq (\frac43 - \epsilon) n - 2$. From there, we obtain a TSP tour of $G$ with the improved guarantee.

Our analysis allows us to set $\epsilon$ as $1/61236>0.000016$ for cubic 2-connected graphs. It is worth noting here that by adding extra hypothesis it is possible to improve this constant. For instance, the case of cubic 2-connected bipartite graphs is interesting. This type of graphs is actually 3-edge colorable; therefore, by taking the complements of the three perfect matchings induced by the coloring, we obtain a cycle cover distribution whose support has only 3 cycle covers making the problem much easier to analyze. In fact, by slightly relaxing the simplification phase we can impose that the resulting graph is still bipartite (but allowing certain type of chorded hexagons) and thus, the operation consisting on the alternation of 5-cycles will never occur. It is possible to show (see Larr\'e's Master Thesis~\cite{L12}) that this modified algorithm yields  a tour of length at most $(4/3-1/108)n-2$ for any cubic 2-connected bipartite graph.

\section{Barnette graphs}\label{section:Barnette}

Barnette~\cite{B69} conjectured that every cubic, bipartite, 3-connected planar graph is Hamiltonian. More than forty years later and despite considerable effort, Barnette's conjecture is still not settled. This motivates the definition of a {\em Barnette graph} as a cubic, bipartite, 3-connected planar graph. Even though we are not able to prove or disprove Barnette's conjecure, in this section we show that this type of graphs admit short tours.

Recall that a tour of a graph $G$ is simply a spanning Eulerian subgraph where
every edge appears at most twice. The main idea for obtaining short tours is to
find a cycle cover $\C$ of $G$ having small number of cycles. The tour
$T(\C)$ obtained by taking the union of the edges of $\C$ and a doubled
spanning tree of the multigraph obtained by contracting each cycle of $\C$
in $G$ has length $n+2|\C|-2$, where $|\C|$ is the number of cycles in $\C$.

Let $G$ be a Barnette graph on $n$ vertices. By 3-connectedness, $G$ has a
unique embedding on the sphere up to isomorphism~\cite{W33}, therefore its set of faces is
well-determined. Furthermore, it has a unique planar dual $G^*$, which is an
Eulerian planar triangulation. As Eulerian planar triangulations are known to be
3-colorable (see, e.g., \cite{TW11}), Barnette
graphs are 3-face colorable. Furthermore, finding such coloring can be done in
polynomial time.

We can use this to easily get a tour of $G$ of length at most $4n/3$ as follows.
Denote the vertex, edge and face sets of $G$ as $V=V(G)$, $E=E(G)$ and $F=F(G)$
respectively and let $c\colon F\to \{1,2,3\}$ be a proper 3-face coloring
of~$G$. Let $F(i)$ be the set of faces of color $i$. Since
the graph is cubic, $|E|=3n/2$, and by Euler's formula, $|F|=2+|E|-|V|=(n+4)/2$.
This means that there is a color $i$ such that $|F(i)| \leq (n+4)/6$.
Since $F(i)$ is a cycle cover, the tour $T(F(i))$ obtained as before has
length $n+2|F(i)|-2 \leq (4n-2)/3$.

We will devise an algorithm that find a short cycle covers of Barnette graphs by
performing certain local operations to reduce the number of cycles of the three
cycle covers given by $F(1)$, $F(2)$ and $F(3)$. Recall that an even cycle $C_0$
is \emph{alternating} for a cycle cycle cover $\C$ of $G$ if the edges
of $C_0$ alternate between edges inside $\C$ and edges outside $\C$. If $C_0$ is
an alternating cycle of $\C$ we can define a new cycle cover $\C \triangle C_0$
whose edge set is the symmetric difference between the edges of $\C$ and those
of $C_0$.

The next procedure constructs cycle covers $\C(i)$, for each $i \in
\{1,2,3\}$. Initialize $\C(i)$ as $F(j)$ for some $j\neq i$, for example $j=i+1
\pmod 3$. As we will see later, each face of $F(i)$ is alternating for $\C(i)$
at every moment of the procedure. Iteratively, check if there is a face $f$ in
$F(i)$ such that $\C(i) \triangle f$ has fewer cycles than $\C(i)$. If so,
replace $\C(i)$ by the improved cover. Do this step until no improvement is
possible to obtain the desired cover. Let $\C$ be the cycle cover of fewer cycles among the three covers found. By returning the tour associated to $\C$ we obtain Algorithm~\ref{alg:tsp_barnette}
below.

\begin{algorithm}
\caption{to find a tour on a Barnette graph $G=(V,E)$.}
\begin{algorithmic}[1]
\STATE Find a 3-face coloring of $G$ with colors in $\{1,2,3\}$.
\FOR{each $i \in \{1,2,3\}$}
\STATE $\C(i) \gets F(j)$, for $j=(i+1)\pmod 3$.
\WHILE{there is a face $f \in F(i)$ such that $|\C(i) \triangle f|<|\C(i)|$}
\STATE{$\C(i) \gets \C(i) \triangle f$.}
\ENDWHILE
\ENDFOR
\STATE Let $\C$ be the cycle cover of smalles cardinality among all
$\C(i)$.
\STATE Return the tour associated to $\C$.
\end{algorithmic}
\label{alg:tsp_barnette}
\end{algorithm}
To analyze the algorithm it will be useful to extend the initial
face-coloring to an edge-coloring of $G$, by assigning to each edge $e$ the
color in $\{1,2,3\}$ that is not present in both incident faces of $e$. Denote
as $E(i)$ the set of edges of color~$i$. Then $E(1)$, $E(2)$ and $E(3)$ are
disjoint perfect matchings and furthermore $E(i)\cup E(j)$ are exactly the edges
of $F(k)$, for $\{i,j,k\}=\{1,2,3\}$. Lemma~\ref{lem:contain} below implies,
in particular, that the algorithm is correct.

\begin{lem} \label{lem:contain} In every iteration of
Algorithm~\ref{alg:tsp_barnette}, $\C(i)$ is a cycle cover of $G$
containing $E(i)$ and every face $f \in F(i)$ is alternating for
$\C(i)$.\end{lem}
\begin{proof}
We proceed by induction. Observe that $\C(i)$ equals $F(j)$ for $j\neq i$ in
the beginning so it contains $E(i)$. Suppose the lemma holds at the
beginning of an iteration and let $f$ be a face of $F(i)$, so $f$ has no
edges of color $i$. As $\C(i)$ contains $E(i)$, and every vertex $v$ of $f$
has degree 2 in $\C(i)$ we conclude that $f$ is alternating for $\C(i)$
and thus, $\C(i) \triangle f$ is a cover containing $E(i)$.
\end{proof}

\begin{lem}
\label{lem:cycle}
Let $\C(i)$ be the cycle cover obtained at the end of the while-loop in
Algorithm~\ref{alg:tsp_barnette} and let $f$ be a face of length $2k$ in
$F(i)$, then $f$ intersects at most $\lfloor \frac{k+1}{2}\rfloor$ cycles of
$\C(i)$.
\end{lem}
\begin{proof}
Let $\C'$ be the collection of cycles in $\C(i)$ intersecting $f$ and let $H$
be the subgraph of $G$ whose edge set is the union of the edges of $\C'$ and
those of $f$. Using the planar embedding of $G$ having $f$ as the outer
face, we can see that $H$ consists of the outer cycle $f$ and a collection of
non-crossing paths on the inside connecting non-adjacent vertices of $f$ as in
Figure~\ref{fig:alternating2k}.

\begin{figure}[htbp]
\centering
\includegraphics[width=.3\textwidth]{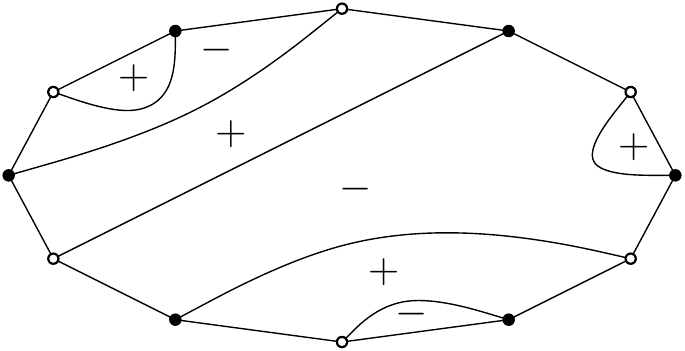}
\caption{The graph $H$ in the proof of Lemma~\ref{lem:cycle}.}
\label{fig:alternating2k}
\end{figure}

Label with a plus ($+$) sign the regions of $H$ bounded by cycles
in $\C(i)$ and by a minus ($-$) sign the rest of the regions except for the outer
face, thus $|\C'|$ equals the number of $+$ regions. By construction the number
of $+$ regions must be at most that of $-$ regions, as otherwise replacing
$\C(i)$ by $\C(i)\cap f$ would decrease the number of cycles in $\C(i)$.
But since $f$ has length $2k$, the total number of regions labeled $+$ and $-$
is $k+1$. Since $|C'|$ is an integer, we get $|C'| \leq \lfloor (k+1)/2
\rfloor$.
\end{proof}

With this lemma we can bound the size of the tour returned. For $i \in
\{1,2,3\}$ and $k\in \N$ define $F_k(i)$ as the
set of faces of color $i$ and length $k$, and $F_k$ as the total number of
faces of length $k$.

\begin{lem}
\label{lemma:firstbound}
 Let $\C(i)$ be the cycle cover obtained at the end of the while-loop in
Algorithm~\ref{alg:tsp_barnette}. Then
\begin{align*}
 |\C(i)| &\leq  1 + \sum_{k=2}^{\infty} \Big\lfloor \frac{k-1}{2} \Big\rfloor
|F_{2k}(i)| = 1 + |F_6(i)|+  |F_8(i)|+  2|F_{10}(i)|+ 2|F_{12}(i)|+ \dots
\end{align*}
\end{lem}
\begin{proof} Let $G'$ be the graph $G$ restricted to the edges of $\C(i)$. Let $H$ be the
connected Eulerian multigraph obtained by contracting in $G'$ all the faces in
$F(i)$ to vertices. Observe that the edge set of $H$ is exactly $E(i)$.
One by one, uncontract each face $f$ in $H$. In each step we obtain an Eulerian
graph which may have more connected components than before. We can estimate the
number of cycles in $\C(i)$ as 1 (the original component in $H$) plus the the
increase on the number of connected components on each step of this
procedure.

Consider the graph $H$ immediately before expanding a face $f$. Let $H_f$
be the connected component of $H$ containing the vertex associated to $f$. If
$f$ has length $2k$, then, after expanding $f$, $H_f$ splits into at most
$\lfloor \frac{k+1}{2}\rfloor$ connected components. This follows since after
expanding all the cycles of $F(i)$, $H_f$ is split into at most that many
components by Lemma~\ref{lem:cycle}. But then, expanding $f$ increases the
number of connected components of $H$ by at most $\lfloor \frac{k+1}{2}\rfloor-1
= \lfloor \frac{k-1}{2}\rfloor.$
\end{proof}

The previous bound is not enough to conclude the analysis. Fortunately we can
find another bound that will be useful.

\begin{lem}
\label{lemma:secondbound}
 Let $\C(i)$ be the cycle cover obtained at the end of the while-loop in
Algorithm~\ref{alg:tsp_barnette}. If $|C(i)|\neq 1$, then
\begin{align*}
 3|\C(i)| &\leq  |\C(i)\cap F_4| + \sum_{k=2}^{\infty} \Big\lfloor
\frac{k+1}{2} \Big\rfloor |F_{2k}(i)|\\
&= |\C(i)\cap F_4| + |F_4(i)| + 2|F_6(i)|+  2|F_8(i)|+  3|F_{10}(i)|+
\dots
\end{align*}
\end{lem}
\begin{proof}
Since $\C(i)$ contain $E(i)$, every cycle $C \in
\C(i)$ must intersect at least two faces in $\F(i)$. In particular,
every 4-cycle in $\C(i)$ intersects exactly two faces in $\F(i)$. Consider now a
cycle $C\in C(i)$ of length at least 6.

We claim that $C$ must intersect three of more faces of $\F(i)$. If this was not
the case then $C$ intersects exactly two faces $f_1$ and $f_2$. Then the edges
of $C$ must alternate as one edge of $f_1$, then one edge in $E(i)$ crossing
between the faces, then one edge in $f_2$ and one edge in $E(i)$ crossing
between the faces. In particular, the length of $C$ must be divisible by 4. Let
$v_0$, $v_1$, \ldots, $v_{4k-1}$ be the vertices of $C$. By the previous
observation we can assume that for every $1\leq \ell\leq k$, $v_{4\ell}v_{4\ell
+1}$ is an edge of $f_1$, $v_{4\ell+2}v_{4\ell+3}$ is an edge of $f_2$ and the
rest of the edges of $C$ are crossing between faces, as depicted in
Figure~\ref{fig:2faces}.

\begin{figure}[t]
\centering
\includegraphics{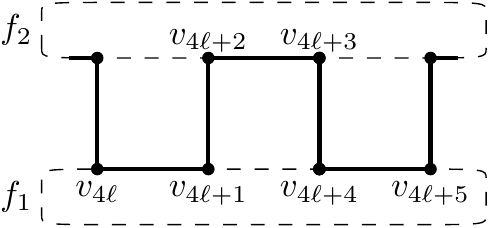}
\caption{A cycle $C$ intersecting two faces $f_1$ and $f_2$.}
\label{fig:2faces}
\end{figure}

Suppose that there exist a vertex in $f_1$ not contained in $C$, then
there must be an index $\ell$ such that $v$ is in the path $P$ from
$v_{4\ell +1}$ to $v_{4\ell +4}$ in $f_1$, that is internally disjoint with
$C$. Since $f_1$ is a face of $G$, we deduce that removing
$v_{4\ell+1}$ and $v_{4\ell+4}$ from $G$ disconnects the graph. But this
contradicts the 3-connectedness of $G$. Therefore, $C$ contains all the
vertices in $f_1$ and, by an analogous argument, all the vertices of $f_2$.
Since the graph is cubic and connected, the only possibility is that $C$
contains all the vertices of $G$. But then $|\C(i)|=1$ which
contradicts the hypothesis of the lemma. Therefore, we have proved the every
cycle in $\C(i)$ of length at least six intersects at least 3 faces.

Define the set $$J(i)=\{(C,H) \in \C(i)\times \F(i)\colon \text{ cycle $C$
intersects face $H$ } \}.$$
By the previous discussion we have,
\begin{align*}
|J(i)| &\geq 2 |\C(i)\cap F_4| + 3 |\C(i)\setminus F_4| = 3|\C(i)| -
|\C(i)\cap F_4|.
\end{align*}
By Lemma~\ref{lem:cycle} we have
\begin{align*}
|J(i)| &\leq \sum_{k=2}^{\infty} \Big\lfloor \frac{k+1}{2} \Big\rfloor
|F_{2k}(i)|.
\end{align*}
Combining the last two inequalities we get the desired result.
\end{proof}

Now we have all the ingredients to bound the size of the cycle cover returned by
the algorithm.

\begin{lem}\label{lema:boundcycle}
Let $\C$ be the cycle cover computed
by Algorithm~\ref{alg:tsp_barnette}. Then
$$|\C| \leq \min\left\{\frac{n+4}{6} - \frac{|F_4|}{6},
\frac{(n+1)}{9}+\frac{|F_4|}{6}\right\}.$$
\end{lem}

\begin{proof}
 First we need a bound on a quantity related to previous lemmas.
Let $\alpha$ be $\sum_{k=2}^{\infty} \Big\lfloor \frac{k-1}{2}
\Big\rfloor |F_{2k}| $. We claim that  $\alpha \leq
\frac{1}{2}\left(n-2-|F_4|\right).$

Since $|F|=(n+4)/3$, the claim above is equivalent to proving that
$2|F_4|+2|F|+4\alpha \leq 3n$. Note that $2|F_4|+2|F| +
4\alpha$ equals
\begin{alignat*}{7}
&~~&  2|F_4|    &    &         &  &          &  &           &  & &\\
&+~&  2|F_4|    &+~  & 2|F_6|  &+~& 2|F_8|   &+~& 2|F_{10}| &+~& 2|F_{12}| &+~&
2|F_{14}| &+~\cdots \\
&+~&            &+~  & 4|F_6|  &+~& 4|F_8|   &+~& 8|F_{10}| &+~& 8|F_{12}| &+~&
12|F_{14}|&+~\cdots \\
&=~& 4|F_4|     &+~  & 6|F_6|  &+~& 6|F_8|   &+~& 8|F_{10}| &+~& 8|F_{12}|
&+~&
10|F_{14}|&+~\cdots,
\end{alignat*}
which is upper bounded by $\sum_{k=2}^\infty 2k|F_{2k}|$. This quantity is
the sum of the length of all the faces in $G$. As each vertex is in
three faces, this quantity is at most $3n$, which proves the claim.

By Lemma~\ref{lemma:firstbound} we have,
\[
|\C| \leq \frac{1}{3}\sum_{i=1}^3 |\C(i)| \leq 1+
\frac{\alpha}{3} \leq 1+
\frac{1}{6}(n-2-|F_4|) = \frac{n+4}{6}-\frac{|F_4|}{6}.
\]

On the other hand, using that $|F|=(n+4)/2$ and Lemma~\ref{lemma:secondbound},
we get
\begin{align*}
9|\C| &\leq 3\sum_{i=1}^3|\C(i)| \leq \sum_{i=1}^3|\C(i)\cap F_4| + \alpha +
|F| \leq \sum_{i=1}^3|\C(i)\cap F_4| + \frac{1}{2}(2n+2-|F_4|).
\end{align*}
Note that each cycle in $F_4$ avoids one color, therefore it can only appear in
two cycle covers $\C(i)$, i.e. $\sum_{i=1}^3|\C(i)\cap F_4| \leq 2|F_4|$. From
here we get that
\begin{align*}
|\C| &\leq \frac{1}{9}\left( 2|F_4| + n +1 - |F_4|/2\right) = \frac{n+1}{9} +
\frac{|F_4|}{6}.\; \qedhere \end{align*}
\end{proof}

The previous lemma implies the main result of this section.

\begin{thm}
Let $\C$ be the cycle cover computed by Algorithm~\ref{alg:tsp_barnette} and
$T$ be the tour returned. Then $|\C| \leq \frac{5n+14}{36}$, and the length of
$T$ is at most $|T|=\frac{23n-22}{18} =
\left(\frac{4}{3}-\frac{1}{18}\right)n-\frac{11}{9}$. In particular every
Barnette graph admits a tour of length $|T|$.
\end{thm}
\begin{proof}
The expression on the right hand side of Lemma~\ref{lema:boundcycle} is
maximized when $|F_4|=\frac{n+10}{6}$, for which it attains a value of
$\frac{5n+14}{36}$. Therefore, for every value of $|F_4|$, this quantity
is an upper bound of $|\C|$. To conclude, we just use that the length of $T$ is
$n+2(|\C|-1)$.
\end{proof}

\section{2-connected cubic graphs: Simplification phase}\label{section:simplification}

We now go back to general 2-connected cubic graphs. Our algorithm starts by reducing the input graph $G$ to a simpler 2-connected cubic graph $H$
which does not contain a cycle of length six with
one or more chords as subgraph. In addition our reduction satisfies that if $H$ has a TSP tour of length at most $(4/3-\epsilon)|V(H)|-2$ then $G$ has
a TSP tour of length at most $(4/3-\epsilon)|V(G)|-2$, where $V(H)$ and $V(G)$ denote the vertex sets of $H$ and $G$ respectively.  We will use the
notation $\chi^F \in \{0,1\}^E$ of $F\subset E$, to denote the \emph{incidence vector} of $F$ ($\chi^F_e=1$ if $e\in F$, and $0$ otherwise).

\medskip
\noindent \emph{Reduction 1:}  Let $\gamma$  be a 6-cycle having two chords and let $G[V(\gamma)]$ be the subgraph induced
by the set of vertices contained  in $\gamma$, and let $v_1$ and $v_2$ be the two vertices connecting $\gamma$ to the rest of $G$.
Our reduction replaces $G[V(\gamma)]$  by a 4-cycle with a chord (shown in Figure  \ref{fig:diamante}), identifying $v_1$
and $v_2$ with the vertices of degree 2 in the chorded 4-cycle. This procedure in particular removes the {\em p-rainbow} structure in Boyd et al.
\cite{BSSS11}.
\begin{figure}[t]
\centering
\includegraphics[width=.2\textwidth]{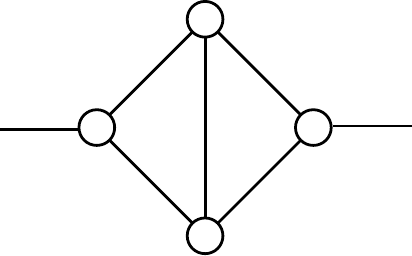}
\caption{A 4-cycle with a chord.}
\label{fig:diamante}
\end{figure}

The second step is to consider 6-cycles having only one chord.
Let $\gamma$ be such a cycle and let $G[V(\gamma)]$ be the subgraph induced
by the set of vertices contained  in $\gamma$.
Consider the four edges $e_1$, $e_2$, $e_3$ and $e_4$ connecting $\gamma$ to the rest of $G$.
Letting $w_i$ be the endpoint of $e_i$ outside $\gamma$, we distinguish the following three reductions according
to three different cases.

\medskip
\noindent \emph{Reduction 2:}  If only two of the $w_i$'s are
distinct, we proceed as in the previous case (Reduction 1),
that is, replacing $G[V(\gamma)]$ by a chorded 4-cycle.

\medskip
\noindent \emph{Reduction 3:}  If three of the $w_i$'s are distinct we replace the 7
vertex structure formed by $\gamma$ and the $w_i$ adjacent to two vertices in $\gamma$ by a triangle (3-cycle),
identifying the degree two vertices in the structure with those in the triangle.
Figure \ref{fig:reduccion_hex_1cuerd_central} shows an example of this reduction
in the specific case that $\gamma$ has a chord connecting symmetrically opposite vertices.
\begin{figure}[htbp]
\centering
\includegraphics[width=.4\textwidth]{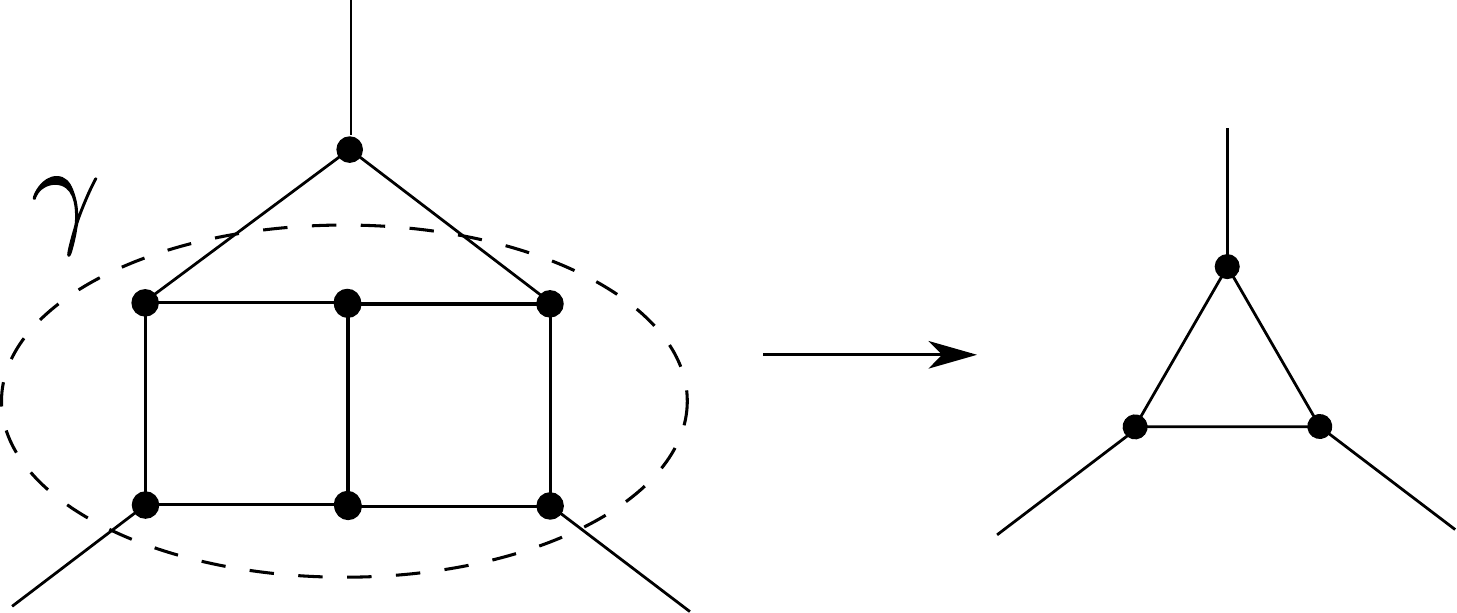}
\caption{Reduction 3  in the case that $\gamma$ has a chord connecting symmetrically opposite vertices.}
\label{fig:reduccion_hex_1cuerd_central}
\end{figure}

\noindent \emph{Reduction 4:}  The final case is when all four $w_i$'s are distinct. Assume, without loss of generality that the $w_i's$ are indexed
in the
cyclic order induced by $\gamma$. In this case we replace
$\gamma$ by an edge $e$ and  we either connect $w_1$, $w_2$ to one
endpoint of $e$ and $w_3$, $w_4$ to the other, or we connect $w_1$, $w_4$ to one endpoint and $w_2$, $w_3$ to the other. The previous choice can
always be made so that in the reduced graph, $e$ is not a bridge, as the following lemma shows.

\begin{lem}
\label{lemma:red4}
 Let $\gamma$ be a chorded 6-cycle considered in Reduction 4.
 Then, the edges in $G[V(\gamma)]$ do not simultaneously contain a cut of $G$
 separating the vertex sets $\{w_1, w_2\}$ and $\{w_3, w_4\}$,
 and a second cut separating the vertex sets $\{w_2, w_3\}$ and $\{w_1, w_4\}$.
\end{lem}
\begin{proof}
 By contradiction: suppose that the edge set of $G[V(\gamma)]$ contains two cuts of $G$ such that the first cut separates the vertex sets $\{w_1,
w_2\}$ and $\{w_3, w_4\}$, and the second cut separates the vertex sets $\{w_2, w_3\}$ and $\{w_1, w_4\}$.
By deleting the vertices and edges of $\gamma$, we split $G$ into four connected components, one containing each $w_i$; thus all paths connecting
$w_1$ and $w_3$ in $G$ must use the edge $e_1$, contradicting the 2-connectivity of $G$.
\end{proof}

Note that all the above reduction steps strictly decrease the number of vertices in the graph while keep the graph cubic and 2-connected,
and that each step requires only polynomial time.
Thus  only a linear number of polynomial time steps needs to be done to obtain a cubic and 2-connected reduced graph
$H$ which does not contain  6-cycles with one or more chords.
The following result shows that any TSP tour in the reduced graph $H$ of length at most $(4/3-\epsilon)|V(H)|-2$
can be turned into a TSP tour in the original graph $G$ of length at most $(4/3-\epsilon)|V(G)|-2$.

\begin{prop}
\label{proposition:conserva_aprox}
 Let $T'$ a TSP tour in the reduced graph $H$ of length at most $\alpha |V(H)|-2$,
 with $5/4 \leq \alpha \leq 4/3$. Then, a TSP tour $T$ can be
 constructed  in the original graph $G$  in polynomial time, such that the length of $T$
 is at most  $\alpha |V(G)|-2$.
\end{prop}
\begin{proof} We distinguish certain cases, depending on what reduction was performed over the graph.
\begin{itemize}[leftmargin=*]
  \item Case of Reduction 1 or 2: in this case the result is a consequence of the following lemma.
\end{itemize}
\begin{lem}
	\label{lemma:2corte_diamante}
	Let $G=(V,E)$ be a graph and $U \subset V$ such that the cut $\delta(U)$ has only two elements, say $\delta(U)=\{e_1, e_2\}$.
	Let $v,w \in U$ be the two end vertices of $e_1$ and  $e_2$ in $U$, respectively.
	Let us suppose that the subgraph $G[U]$ is Hamiltonian and contains a Hamiltonian path connecting $v$ and $w$.
	Let $H$ be the graph resulted from replacing the subgraph $G[U]$ by a chorded 4-cycle $D$
	and let $T'$ be a TSP tour in $H$. Then, there exists a TSP tour $T$ in $G$
	with length $|T|\leq|T'|+|U|-4$.
\end{lem}
\begin{proof}[Proof of Lemma \ref{lemma:2corte_diamante}]
	Let  $\chi^P$ be the incidence vector of some Hamiltonian path $P$ connecting $v$ and $w$, and
	let $\chi^C$ be the incidence vector of some Hamiltonian cycle $C$ in $G[U]$.
	Let us  denote by $x'$ the vector $\chi^{T'}$ (the incidence vector of $T'$).
	We are going to extend the TSP tour $T'$ to the original graph $G$ depending on
	the value of $x'$ on edges $e_1$ and $e_2$.
	We know that $x'(\{e_1, e_2\}):=x'(e_1)+x'(e_2)$ must be an even number,
	since $T'$ is a TSP tour. Considering that $x'$ takes values over $\{0,1,2\}$, we have that all
	the possible cases for the values of $x'(e_1)$ and $x'(e_2)$ are as follows.

\begin{itemize}
\item Case $x'(e_1)=1$ and $x'(e_2)=1$: in this case there is a path in the chorded 4-cycle $D$
and in $T'$ of length 3 connecting $v$ and $w$. Then, we define the vector $x$ as
$$x(e)=
\begin{cases}
 \chi^P(e) & \text{, if } e \in E(G[U]), \\
 x'(e)    & \text{, if } e \in E(G)\setminus E(G[U]]).
\end{cases}
$$

\item Case $x'(e_1)=2$ and $x'(e_2)=0$ (or the symmetrical case): in this case there is a
4-cycle in $D$ and in $T'$. Then, we define the vector $x$ as
 $$x(e)=
\begin{cases}
 \chi^C(e) & \text{, if } e \in E(G[U]), \\
 x'(e)    & \text{, if } e \in E(G)\setminus E(G[U]).
\end{cases}
$$

\item Case $x'(e_1)=2$ and $x'(e_2)=2$: in this case we redefine $x'(e_2)=0$ and then we define
the vector $x$ as the previous case.
\end{itemize}

In any case $x$ is the incident vector of a TSP tour $T$ in $G$, of length $|T|\leq|T'|+|U|-4$.
\end{proof}

It is straightforward to verify that both Reduction 1 and 2 satisfy the hypothesis of Lemma \ref{lemma:2corte_diamante}.
In the case of Reduction 1, the vertex set of the replaced structure has size $|U|=6$,
and in the case of Reduction 2, the vertex set of the replaced structure has size $|U|=8$.
Whatever the case, we have that $|V(G)|=|V(H)|+|U|-4$, and then
\begin{eqnarray*}
 |T| &\leq & |T'|+|U|-4 \\
     &\leq \alpha|V(H)|-2 +|U|-4 \\
     &= & \alpha|V(G)|-2-(\alpha-1)(|V(G)|-|V(H)|) \\
     &\leq & \alpha|V(G)|-2,
\end{eqnarray*}
where the first inequality holds by Lemma \ref{lemma:2corte_diamante}.

\begin{itemize}[leftmargin=*]
\item Case of Reduction 3:  in this case the result is a  consequence of the following lemma.
\end{itemize}
\begin{lem}
\label{lemma:3corte_triangulo}
	Let $G=(V,E)$ a graph and $U \subset V$ such that $|U|=7$
	and the cut $\delta(U)$ has only three elements, say $\delta(U)=\{e_1, e_2, e_3\}$.
	Let $v_1, v_2, v_3 \in U$ be the three end vertices of  $e_1$, $e_2$ and $e_3$ in $U$, respectively.
	Let us suppose that the subgraph $G[U]$ contains a cycle $C$ of length at most 8
	which visits every vertex of $U$, and for every pair of vertices $v,w \in \{v_1, v_2, v_3\}$
	there exists a path $P(v,w)$ of length at most 7 which visits every vertex of $U$. 	
	Let $H$ be the graph resulted from replacing the subgraph $G[U]$ by a triangle
	and let $T'$ be a TSP tour in $H$. Then, there exists a TSP tour $T$ in $G$
	with length $|T|\leq|T'|+5$.
\end{lem}
\begin{proof}[Proof of Lemma \ref{lemma:3corte_triangulo}. ]
	Let $\chi^P(v,w)$ be the incidence vector of path $P(v,W)$ and  $\chi^C$ be the incidence vector of $C$.
	Let us denote by $x'$ the vector $\chi^{T'}$ (the incidence vector of $T'$).
	We are going to extend the TSP tour $T'$ to the original graph $G$ depending on
	the value of $x'$ on edges $e_1$, $e_2$ and $e_3$.
	We know that $x'(\{e_1, e_2, e_3\}):=x'(e_1)+x'(e_2)+x'(e_3)$ must be an even number,
	since $T'$ is a TSP tour. Considering that $x'$ takes values over $\{0,1,2\}$, we have that all
	the possible cases for the values of $x'(e_1)$,  $x'(e_2)$ and $x'(e_3)$ are as follows.
\begin{itemize}
\item Case $x'(e_1)=2$ and $x'(e_2)=x'(e_3)=0$ (or another possible permutation):
in this case we define the vector $x$ as
$$x(e)=
\begin{cases}
 \chi^C(e) & \text{, if } e \in E(G[U]), \\
 x'(e)    & \text{, if } e \in E(G)\setminus E(G[U]]).
\end{cases}
$$

\item Case $x'(e_1)=2$, $x'(e_2) \in \{0,2\}$ and $x'(e_3) \in \{0,2\}$ (or another possible permutation):
first we redefine $x'(e_2)=x'(e_3)=0$ and then we define the vector $x$ as the previous case.

\item Case $x'(e_1)=x'(e_2)=1$ and $x'(e_3)=0$ (or another possible permutation):
in this case we define the vector $x$ as
$$x(e)=
\begin{cases}
 \chi^{P(v_1,v_2)}(e) & \text{, if } e \in E(G[U]), \\
 x'(e)    & \text{, if } e \in E(G)\setminus E(G[U]]).
\end{cases}
$$

\item Case $x'(e_1)=x'(e_2)=1$ and $x'(e_3)=2$ (or another possible permutation): first we redefine $x'(e_3)=0$
and then we define the vector $x$ as the previous case.
\end{itemize}

Clearly, in any case $x$ is the incident vector of a TSP tour $T$ in $G$.
Since at most 5 edges were necessary to construct $x$ from $x'$, we have that $|T|\leq|T'|+5$.
\end{proof}

To see that all possible structures that are considered in Reduction 3
satisfy the hypothesis of Lemma \ref{lemma:3corte_triangulo}, we only need to check by inspection
that there exist a cycle $C$ of length at most 8 and paths of length at most 7 connecting every pair of vertices,  which
visit every vertex of the structures of Figure \ref{fig:posibles_reduccion3}.
\begin{figure}[htbp]
\centering
\includegraphics[width=.7\textwidth]{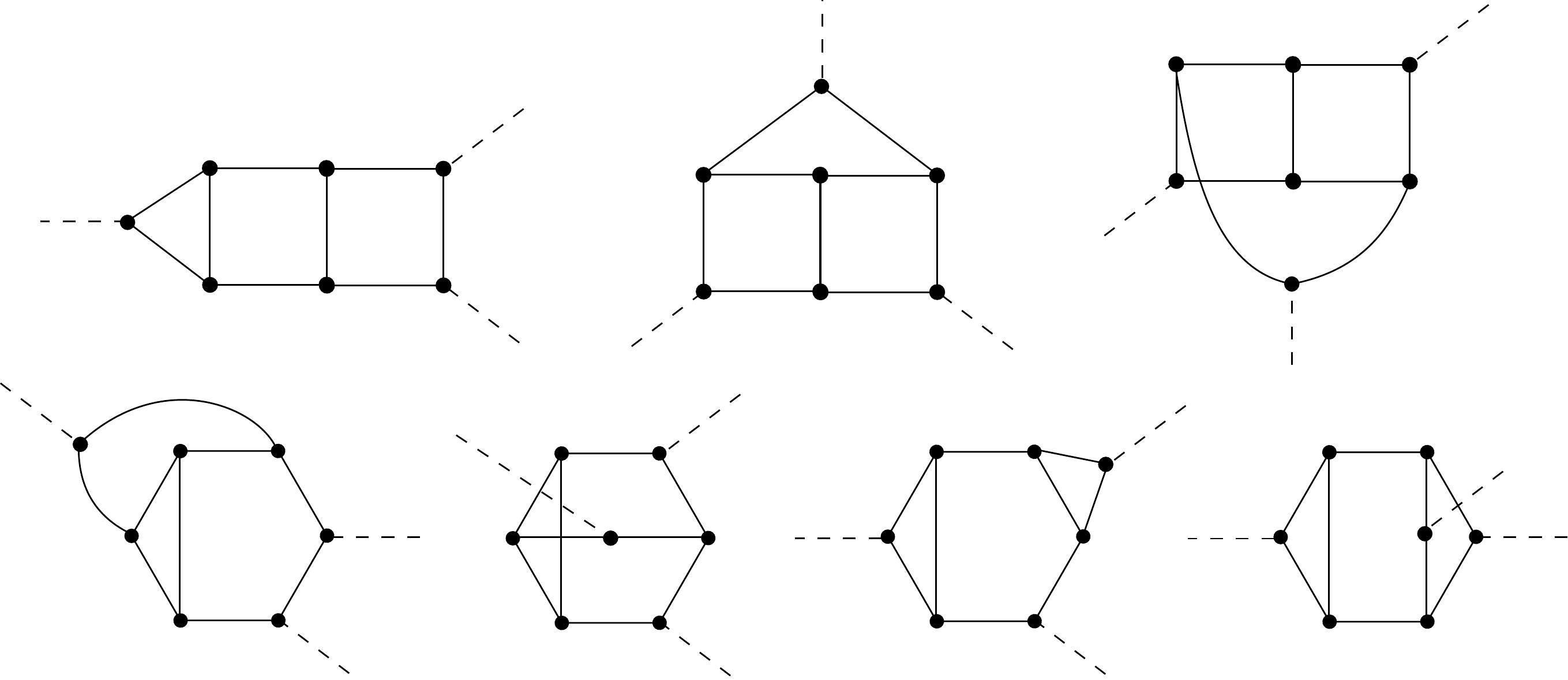}
\caption{All possible structures that are considered in Reduction 3.}
\label{fig:posibles_reduccion3}
\end{figure}

Then, applying the Reduction 3, we can define a TSP tour $T$ in the original graph with length
\begin{eqnarray*}
 |T| &\leq & |T'|+5 \\
     &\leq & \alpha|V(H)|-2 +5 \\
     &= & \alpha|V(G)|-2+5 + \alpha(|V(H)|-|V(G)|) \\
     &= & \alpha|V(G)|-2+(5-  \alpha 4) \\
     &\leq & \alpha|V(G)|-2,
\end{eqnarray*}
where the first inequality holds by Lemma \ref{lemma:3corte_triangulo} and the last one since $5/4 \leq \alpha$.

\begin{itemize}[leftmargin=*]
\item Case of Reduction 4: in this case it is easy to construct a TSP tour $T$ in the original graph $G$ with length $|T|\leq|T'|+4$.
Then
\begin{eqnarray*}
 |T| &\leq & |T'|+4 \\
     &\leq & \alpha|V(H)|-2 +4 \\
     &= & \alpha|V(G)|-2+4 + \alpha(|V(H)|-|V(G)|) \\
     &= & \alpha|V(G)|-2+(4-  \alpha 4) \\
     &\leq & \alpha|V(G)|-2,
\end{eqnarray*}
where the last inequality holds since $\alpha \geq 1$.
\end{itemize}

Considering this latter case, we finish the proof of Proposition \ref{proposition:conserva_aprox}.
\end{proof}

Finally, note that --as mentioned above--  only a linear number of reduction steps need,  and each step requires only polynomial time, not only to
find the desired structure, but also to recover the TSP tour in the original graph. Thus this graph simplification phase runs in polynomial time.

\section{2-connected cubic graphs: Eulerian subgraph cover phase}\label{section:main}
We say that a matching $M$ is \emph{3-cut perfect} if $M$ is a perfect matching intersecting every 3-cut in exactly one edge.
Boyd et al.~\cite{BSSS11} have shown the following lemma.

\begin{lem}[\cite{BSSS11}] \label{lemma:3perfMatchings}
Let $G=(V,E)$ be a $2$-connected cubic graph.
Then, the vector $\frac{1}{3}\chi^E$ can be expressed as a convex combination of incident vectors of $3$-cut perfect matchings of $G$.
This is, there are 3-cut perfect matchings $\{M_i\}_{i=1}^k$ and positive real numbers $\lambda_1,\lambda_2,\dots,\lambda_k$ such that

\begin{tabular}{lcl}
\begin{minipage}[b]{.3\textwidth}
\begin{equation}
\sum_{i=1}^k \lambda_i =1 \label{eq:sum1}
\end{equation}
\end{minipage} &\qquad\qquad and \qquad\qquad&
\begin{minipage}[b]{.3\textwidth}
\begin{equation}
\frac{1}{3}\chi^E=\sum_{i=1}^k \lambda_i \chi^{M_i}. \label{eq:convexMat}
\end{equation}
\end{minipage}
\end{tabular}

Furthermore, Barahona~\cite{B04} provides an algorithm to find a convex combination of $\frac13 \chi^E$ having $k\leq 7n/2-1$ in
$O(n^6)$ time.
\end{lem}

\medskip
Consider a graph $G$ that is cubic, 2-connected and \emph{reduced}.
That is, no 6-cycle in $G$ has chords.
We also assume that $n\geq 10$ as every cubic 2-connected graph on less than 10 vertices is Hamiltonian.

Let $\{M_i\}_{i=1}^k$ and $\{\lambda_i\}_{i=1}^k$ be the 3-cut matchings and coefficients guaranteed by Lemma~\ref{lemma:3perfMatchings}.
Let $\{\C_i\}_{i=1}^k$ be the family of \emph{cycle covers} associated to the matchings $\{M_i\}_{i=1}^k$.
This is, $\C_i$ is the collection of cycles induced by $E\setminus M_i$.
Since each matching $M_i$ is 3-cut perfect, the corresponding cycle cover $\C_i$ does not contain 3-cycles.
Furthermore every 5-cycle in $\C_i$ is induced (i.e., it has no chord in $G$).

In what follows we define three local operations, (U1), (U2) and (U3) that will be applied iteratively to the current family of covers. Each operation
is aimed to reduce the contribution of each component of the family. We stress here that operations (U2) and (U3) are exactly those used by Boyd et
al., but for reader's convenient we explain them here. We start with operation (U1).

\begin{itemize}
\item[(U1)] Consider a cycle cover $\C$ of the current family. If $C_1$, $C_2$ and $C_3$ are three disjoint cycles of $\C$, that intersect a fixed
6-cycle $C$ of $G$, then we merge them into the simple cycle obtained by taking their symmetric difference with $C$. This is, the new cycle in $V(C_1)
\cup V(C_2) \cup V(C_3)$ having edge set $(E(C_1) \cup E(C_2) \cup E(C_3)) \Delta E(C)$.
\end{itemize}

An example of (U1) is depicted in Figure~\ref{fig:operacion_i}. We apply (U1) as many times as possible to get a new cycle
cover $\{\C^{\UU1}_i\}_{i=1}^k$. Then we apply the next operation.
\begin{figure}[t]
\centering
\includegraphics[width=.35\textwidth]{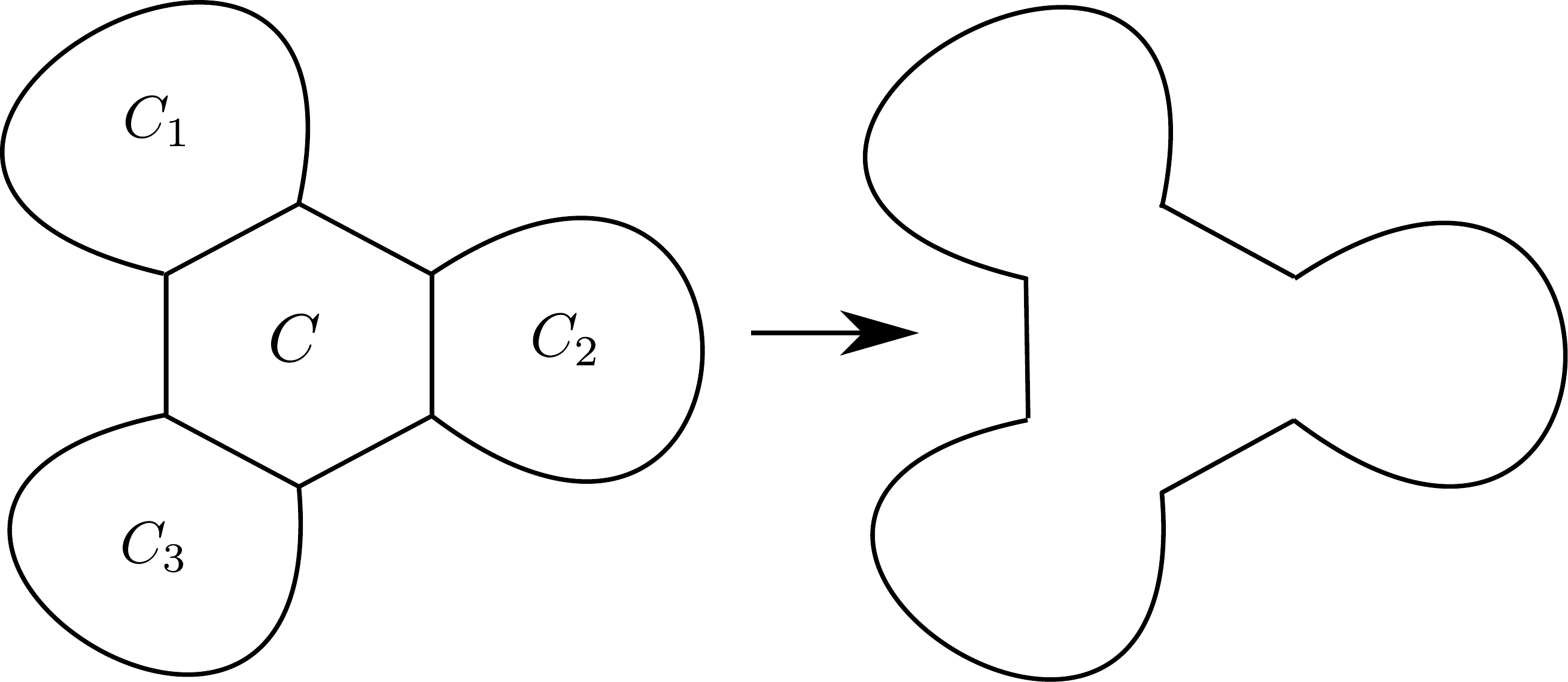}
\caption{Operation (U1).}
\label{fig:operacion_i}
\end{figure}

\begin{itemize}
 \item[(U2)]  Consider a cycle cover $\C$ of the current family. If $C_1$ and $C_2$ are two disjoint cycles of $\C$, that intersect a fixed 4-cycle
$C$ of $G$, then we merge them into a simple cycle obtained by taking their symmetric difference with $C$. This is, the new cycle in $V(C_1) \cup
V(C_2)$ having edge set $(E(C_1) \cup E(C_2)) \Delta E(C)$.
\end{itemize}

We apply operation (U2) as many times as possible to obtain a new cycle cover $\{\C^{\UU2}_i\}_{i=1}^k$ of $G$. The next operation we define
may transform a cycle cover $\C$ of the current family into a Eulerian subgraph cover $\Gamma$, having components that are
not necessarily cycles.

\begin{itemize}
 \item[(U3)] Let $\Gamma$ be an Eulerian subgraph cover of the current family. If $\gamma_1$ and $\gamma_2$ are two components of $\Gamma$, each one
having at least 5 vertices, whose vertex set intersect a fixed 5-cycle $C$ of $G$, then combine them into a single component, by adding at most 1
extra edge.
\end{itemize}

To explain how we combine the components in operation (U3) we need the following two lemmas.

\begin{lem}[\cite{BSSS11}]\label{lemma:merging}
 Let $H_1$ and $H_2$ be two connected Eulerian multi-subgraphs of a cubic graph $G$ having at least two vertices in common and let $H_3$ be the sum of
$H_1$ and $H_2$, i.e., the union of their vertices and the sum of their edges (allowing multiple parallel edges). Then we can remove (at least) two
edges from $H_3$ such that it stays connected and Eulerian.
\end{lem}
\begin{proof}
Let $u$ and $v$ be in both $H_1$ and $H_2$. The edge set of $H_3$ can be partitioned into edge-disjoint $(u,v)$-walks
$P_1$, $P_2$, $P_3$ and $P_4$. Since $u$ has degree 3 in $G$, there must be two parallel edges incident to $u$ that are on different paths, say $e_1
\in P_1$ and $e_2 \in P_2$. If we remove $e_1$ and $e_2$ then the graph stays Eulerian. Moreover, it stays connected since $u$ and $v$ are still
connected by $P_3$ and $P_4$ and every vertex of $P_1$ and $P_2$ is still connected to one of $u$ and $v$.
\end{proof}

\begin{lem}[Similar to an observation in \cite{BSSS11}] \label{lemma:property-P}
If $v$ belongs to a component $\gamma$ of any of the covers $\Gamma$ considered by the algorithm, then at least two of its 3 neighbors are in the same
component.
\end{lem}

\begin{proof}
The lemma holds trivially when $\gamma$ is a cycle. In particular, the lemma holds before the application of operation (U3). As the vertex
set of a component created by operation (U3) is the union of the vertex set of 2 previous components, the lemma also holds after  operation (U3).
\end{proof}

Observe that if $\gamma$ is a component of a cover in the current family, and $C$ is an arbitrary cycle of $G$ containing a vertex of $\gamma$ then,
by the cubicity of $G$ and Lemma~\ref{lemma:property-P}, $C$ and $\gamma$ must share at least two vertices. In particular, if $\gamma_1$ and
$\gamma_2$ are the two components intersecting a 5-cycle $C$ considered by operation (U3), then one of them, say $\gamma_1$, must contain exactly 2
vertices of $C$ and the other one must contain the other 3 vertices (note that they cannot each share 2 vertices, since then a vertex of $C$ would not
be included in the cover). To perform  (U3) we first merge $\gamma_1$ and $C$ using Lemma~\ref{lemma:merging} removing 2 edges, and then we
merge the resulting component with $\gamma_2$, again removing 2 edges. Altogether, we added the 5 edges of $C$ and removed 4 edges.
Finally, we remove 2 edges from each group of triple or quadruple edges that may remain, so that each edge appears at most twice in each component.
Figure~\ref{fig:opiiPent} shows an example of  (U3).
\begin{figure}[t]
\centering
\includegraphics[width=.45\textwidth]{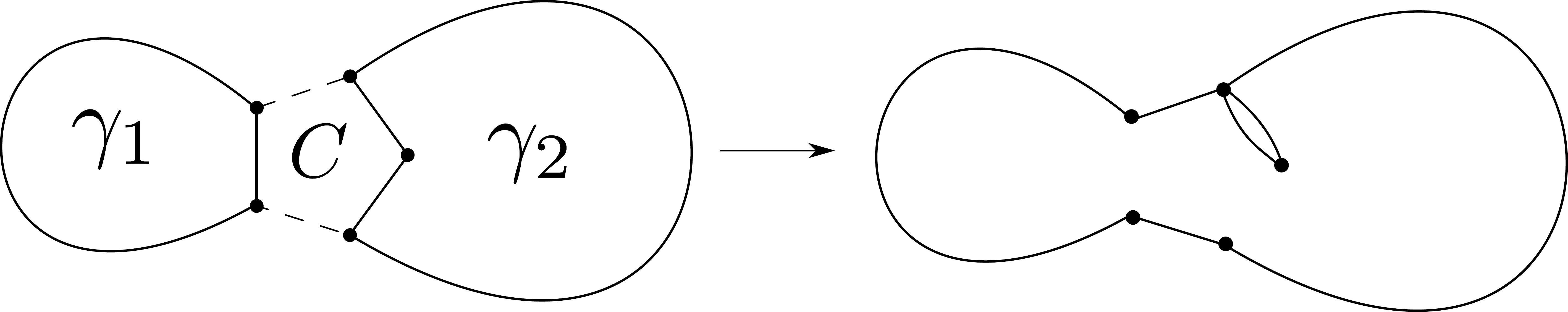}
\caption{Sketch of operation (U3).}
\label{fig:opiiPent}
\end{figure}

\begin{rem}\label{remark:ciclos}
Operation (U3) generates components having at least 10 vertices. Therefore, any component having 9 or fewer vertices must be a cycle. Furthermore, all
the cycles generated by (U1) or (U2) contain at least 10 vertices (this follows from the fact that $G$ is reduced, and so operation (U2)
always involve combine 2 cycles of length at least 5). From here we observe that any component having 9 or fewer vertices must be in the original
cycle cover $\{C_i\}_{i=1}^k$.
\end{rem}

We say that a 4-cycle $C$ with a chord is isolated if the two edges incident to it are not incident to another chorded 4-cycle.
The following is the main result of this section. Before proving it we show it implies the main result of the paper.
\medskip
\begin{prop}[Main Proposition]
Let $\{\Gamma_i\}_{i=1}^k$ be the family of Eulerian subgraph covers at the end of the algorithm (that is, after applying all operations), and let
$z(v)=z_\calD(v)$ be the average contribution of vertex $v$ for the distribution $\calD=\{(\Gamma_i, \lambda_i)\}_{i=1}^k$. Furthermore, let
$\gamma_i$ be the component containing $v$ in $\Gamma_i$ and $\Gamma(v) = \{\gamma_i\}_{i=1}^k$. We have the following.
\begin{enumerate}
\item[(P1)] If $v$ is in an isolated chorded 4-cycle then $z(v)\leq 4/3$.
\item[(P2)] If $v$ is in a non-isolated chorded 4-cycle of $G$ then $z(v)\leq 13/10$.
\item[(P3)] Else, if there is an induced 4-cycle $\gamma \in \Gamma(v)$, then $z(v)\leq 4/3 - 1/60$.
\item[(P4)] Else, if there is an induced 5-cycle $\gamma \in \Gamma(v)$, then $z(v)\leq 4/3 - 1/60$.
\item[(P5)] Else, if there is an induced 6-cycle $\gamma \in \Gamma(v)$, then we have both $z(v) \leq 4/3$ and $\sum_{w \in V(\gamma)} z(w) \leq
6\cdot (4/3 - 1/729)$.
\item[(P6)] In any other case $z(v) \leq 13/10$.
\end{enumerate}
\end{prop}
\medskip

\begin{thm}\label{thm:main}
  Every 2-connected cubic graph $G=(V,E)$ admits a TSP tour of length at most $(4/3-\epsilon)|V|-2$, where $\epsilon=1/61236$. This tour can be
computed in polynomial time.
\end{thm}
\begin{proof}[Proof of Theorem~\ref{thm:main}]

From Section~\ref{section:simplification}, we can assume that $G$ is also reduced and so the Main Proposition holds. Let $B$ be the
union of the vertex sets of all isolated chorded 4-cycles of~$G$. We say a vertex is \emph{bad} if it is in $B$, and \emph{good} otherwise. We claim
that the proportion of bad vertices in $G$ is bounded above by $6/7$. To see this, construct the auxiliary graph $G'$ from $G$ by replacing every
isolated chorded 4-cycle with an edge between its two neighboring vertices. Since $G'$ is cubic, it contains exactly $2|E(G')|/3$ vertices, which are
good in $G$.
Hence, for every bad vertex there are at least $(1/4) \cdot (2/3)  = 1/6$ good ones, proving the claim.

The Main Proposition guarantees that every bad vertex $v$ contributes a quantity $z(v)\leq 4/3$. Now we show that the average contribution of all the
good vertices is at most $(4/3 - \delta)$ for some $\delta$ to be determined. To do this, define $\mathcal{H}=\{\gamma \in \bigcup_i\Gamma_i\colon
|V(\gamma)|=6\}$ as the collection of all 6-cycles appearing in some cover of the final family, and let $H=\bigcup_{\gamma \in \mathcal{H}}V(\gamma)$
be the vertices included in some 6-cycle of $\mathcal{H}$. It is easy to check that $B$ and $H$ are disjoint. Furthermore, the Main Proposition
guarantees that if $v \in V\setminus (B\cup H)$ then $z(v)\leq (4/3 - 1/60)$. So we focus on bounding the contribution of the vertices in $H$.

For every $v \in H$, let $f(v)$ be the number of distinct 6-cycles in $\mathcal{H}$ containing~$v$. Since $G$ is cubic, there is an absolute constant
$K$, such that $f(v)\leq K$. By the main proposition, $z(v)\leq 4/3$ for $v\in H$ and for every $\gamma \in \mathcal{H}$, $\sum_{v\in V(\gamma)} z(v)
\leq 6\cdot (4/3-\epsilon')$, where $\epsilon'=1/729$. Putting all together we have:
{\small \begin{align*}
&\hphantom{=}K\cdot \sum_{v \in H} \left[z(v) - \left(\frac43 - \frac{\epsilon'}{K} \right)\right] = |H|\epsilon' + K \sum_{v\in H}
\left(z(v)-\frac43\right)\\
&\leq 6|\mathcal{H}|\epsilon' + \sum_{v\in H} f(v)\left(z(v)-\frac43\right) = 6|\mathcal{H}|\epsilon' + \sum_{\gamma \in \mathcal{H}} \sum_{v \in
V(\gamma)} \left(z(v)-\frac43\right)\\
&\leq 6|\mathcal{H}|\epsilon' - \sum_{\gamma \in \mathcal{H}} 6\epsilon' = 0.
\end{align*}}
It follows that $\frac1{|H|}\sum_{v \in H} z(v) \leq \left(4/3 - \epsilon'/K\right)$. Since $\epsilon'/K \leq 1/60$, we get
{\small \begin{align*}
\sum_{v \in V} z(v) &\leq \sum_{v \in B} z(v) + \sum_{v \in H} z(v) + \sum_{v \in V\setminus (B\cup H)} z(v)\\
&\leq \frac43 |B| +  \left(\frac43 - \frac{\epsilon'}{K}\right)(|V|-|B|) = |V|\left(\frac43 - \frac{\epsilon'}{7K}\right).
\end{align*}}
We conclude that there is an index $i$ such that $\sum_{v \in V} z_i(v) \leq |V|\left(4/3 - {\epsilon'}/(7K)\right)$. By adding a double spanning tree
of $G/E(\Gamma_i)$ we transform $\Gamma_i$ into a TSP tour $T$ of length $|V|\left(4/3 - {\epsilon'}/(7K)\right)-2$.  Noting that $K\leq 12$
 and $\epsilon'=1/729$ we obtain the desired bound\footnote{Consider
 two edges $e_1$ and $e_2$ adjacent to $v$. Since there is no chorded 6-cycle, if $e_1$ and $e_2$
 are contained in a 4-cycle, then $v$ must be contained in at most one 6-cycle. Otherwise, there are at most four 6-cycles which may contain $e_1$ and
$e_2$.
 Because there are 3 possible pairs of edges, we have $K=12$.}. Clearly, all
operations can be done in polynomial time.
\end{proof}

\subsection{Proof of the Main Proposition}

We start by a lemma, whose proof is the same as that of \cite[Observation 1]{BSSS11}.

\begin{lem}[\cite{BSSS11}]
\label{lema:tecnicoZiv}
 For each vertex $v\in V$, and each $i \in \{1,...,k\}$, the contribution $z_i(v):=z_{\Gamma_i}(v)$ is
\begin{itemize}
 \item[(a)] at most $\frac{h+2}{h}$, where $h=\min\{t,10\}$ and $v$ is on a t-cycle belonging to one of the cycle covers $\C_i$, $\C_i^{\UU1}$ and
$\C_i^{\UU2}$.
 \item[(b)] at most $\frac{13}{10}$ if operation (U3) modified the component containing $v$.
\end{itemize}
\end{lem}

We will also use the following notation in our proof. For any subset $J$ of indices in  $[k]:=\{1,\ldots,k\}$, define $\lambda(J)=\sum_{i \in
J}\lambda_i$.

The proofs of parts (P1) through (P4) are similar to the arguments used by Boyd et al.~\cite{BSSS11} to show that $z(v)\leq 4/3$ when $v$ is a 4-cycle
or 5-cycle. By using the fact that $G$ is reduced (i.e. it contains no chorded 6-cycles) we obtain a better guarantee in (P2), (P3) and (P4).
To prove part (P5) we heavily use the fact that operation (U1) is applied to the initial cycle
cover (recall that this operation was not used in~\cite{BSSS11}).

\subsubsection{Proof of part (P1) of the Main Proposition}
Let $v$ be in some isolated chorded 4-cycle $C$ with $V(C)=\{a,b,u_0,u_1\}$ as in Figure~\ref{fig:p-rainbow_0}.

\begin{figure}[ht]
\centering
\includegraphics[width=.3\textwidth]{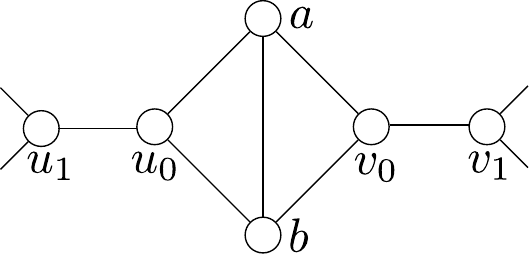}
\caption{A chorded 4-cycle.}
\label{fig:p-rainbow_0}
\end{figure}

For every index $i$, let $C_i$ be the cycle containing $v$ in the initial cycle cover $\C_i$, and let $\C(v)=\{C_i\}_{i=1}^k$.
Consider a cycle $C' \in \C(v)$, and recall that $C'$ cannot be a triangle. If $C'$ does not contain the edge $u_0u_1$, then $C'=C$. Consider now the
case in which $C'$ contains $u_0u_1$. Then we must also have $ab\in E(C')$ and $v_0v_1 \in E(C')$. Since the graph is reduced, $v_1u_1 \notin E$ as
otherwise $u_1-u_0-a-b-v_0-v_1$ would induce a chorded 6-cycle. Hence, the cycle $C'$ cannot be of length 6. It also cannot be of length 7 since then
there would be a 3-cut with 3 matching edges. Therefore, it must be of length at least 8. Using that $\sum_{\{i\colon u_1u_2 \in
M_i\}}\lambda_i=\frac{1}{3}$ and applying Lemma \ref{lema:tecnicoZiv}, we conclude
that $z(v) \leq (1/3 \cdot 6/4 + 2/3 \cdot 10/8)=4/3.$

\subsubsection{Proof of part (P2) of the Main Proposition}
Let $v$ be in some non-isolated chorded 4-cycle $C$ with $V(C)=\{a,b,u_0,u_1\}$ as in Figure~\ref{fig:p-rainbow_0} and recall that $v_1u_1 \not\in E$.
Without loss of generality we can assume that $u_1$ is in a different chorded 4-cycle $D$. Furthermore, assume that $v_1$ is not connected by an edge
to $D$, as this would imply the existence of a bridge in $G$.

Consider, as in the proof of part (P1) a cycle $C' \in \C(v)$. If $C'$ does not contain the edge $u_0u_1$ then $C'=C$. If on the other hand the edge
$v_0v_1$ is in $C'$ then $C'$ must contain all the vertices of both $C$ and $D$. It must also contain $v_1$ and one of its neighbors outside $C\cup
D$. In particular, $C'$ has at least 10 vertices. By Lemma~\ref{lema:tecnicoZiv}, we have
that $z(v) \leq (1/3 \cdot 6/4 + 2/3 \cdot 12/10)=13/10$.

\subsubsection{Proof of part (P3)  of the Main Proposition}

Let $\gamma\in \Gamma(v)$ be an induced 4-cycle containing $v$. By Remark~\ref{remark:ciclos}, $\gamma$ is in some initial cycle cover $\C_i$. Since
the cycle $\gamma$ has no chord, then the four edges incident to it (i.e. those sharing one vertex with $\gamma$) belong to matching $M_i$. This
observation holds not only for $\gamma$ but for any cycle $C^*$ in some initial cycle cover $\C_i$, so we have the following remark.

\begin{rem}
\label{remark:caminoCiclo}
Let $P$ be a path not sharing edges with a cycle $C^*$ belonging to some initial cycle cover $\C_i$. If $P$ connects any two vertices of $C^*$, then
$P$ has length at least~3.
\end{rem}

Furthermore, as the graph is reduced, $\gamma$ does not share exactly one edge with any other 4-cycle (as this would induce a 6-cycle with a chord).
In other words we have the following property.

\begin{rem}\label{remark:caminocuadrado}
Let $P$ be a path not sharing edges with $\gamma$. If $P$ connects any pair of consecutive edges of $\gamma$, then $P$ has length at least 4.
\end{rem}
Define the sets $X_p=\{i\colon |C \cap M_i|=p\}$, for $p=0, 1, 2$ and note that $X_0\cup X_1\cup X_2 =[k]$. Define also $x_p = \lambda(X_p)$, for
$p=0, 1, 2$.

By equation \eqref{eq:sum1}, we have $x_0+x_1+x_2=1$. Also, by applying equation \eqref{eq:convexMat} to the set of 4 edges incident to $\gamma$ we
obtain $4x_0+2x_1=4/3$, which implies that $x_0=1/3-x_1/2$.
Finally, by applying \eqref{eq:convexMat} to the 4 edges inside $\gamma$, we obtain $x_1+2x_2=4/3$, which implies that $x_2=2/3-x_1/2$.

For every $i\in X_0$, the cycle containing $v$ in $\C_i$ is equal to $\gamma$. By Lemma \ref{lema:tecnicoZiv} we obtain $z_i(v) \leq 6/4=3/2$.

Using Remark~\ref{remark:caminocuadrado} we deduce that for every  $i \in X_1$, the cycle containing $v$ in $\C_i$ has length at least 7; therefore,
by Lemma \ref{lema:tecnicoZiv}, we have $z_i(v) \leq 9/7$.

Consider now an index $i \in X_2$. Suppose that $\gamma$ intersects two different cycles of $\C_i$. As each of them has length at least 5 and they
both share one edge with a 4-cycle of $G$ we conclude that both cycles are modified by operation (U1) or (U2). Remark~\ref{remark:ciclos} implies that
$v$ is in a cycle of length at least 10 in $\C_i^{\UU2}$. Using Lemma~\ref{lema:tecnicoZiv} we have $z_i(v)\leq
12/10=6/5$.

The only remaining case is if $\gamma$ is intersected by a single cycle $C$ of $\C_i$. Then, by Remark~\ref{remark:caminocuadrado}, $C$ has length at
least 8. This cycle has length exactly 8 if and only if $\gamma$ belongs to the structure depicted in Figure~\ref{fig:ciclo4prob}.
\begin{figure}[t]
\centering
\includegraphics[width=.35\textwidth]{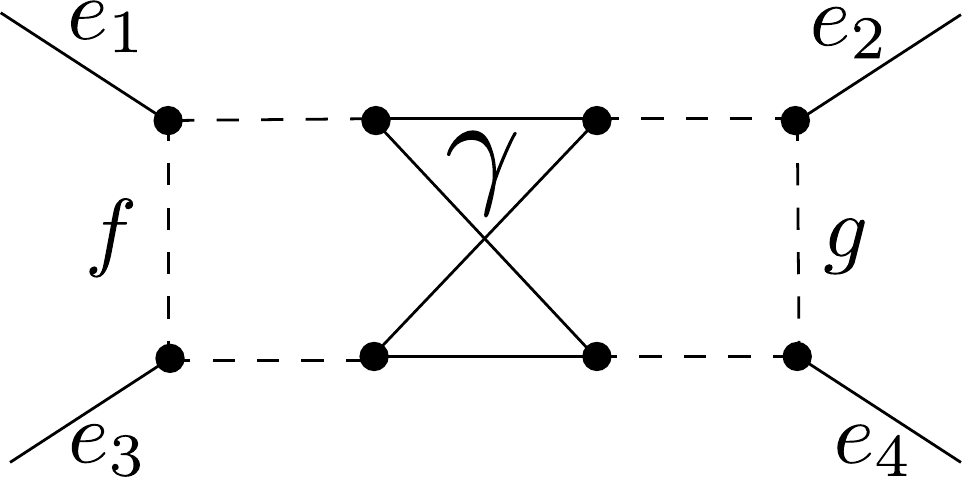}
\caption{4-cycle $\gamma$ intersecting an 8-cycle.}
\label{fig:ciclo4prob}
\end{figure}
Assume for now that no 8-cycle of an initial cover contains the four vertices of $\gamma$. Then, the cycle $C$ in our previous discussion must be of
length at least 9, and by Lemma
\ref{lema:tecnicoZiv}, $z_i(v) \leq \max\{11/9, 6/5\}=11/9$. Putting all together, we obtain
\begin{eqnarray*}
 z(v) &\leq & x_0 3/2 + x_1 9/7 + x_2 11/9 \\
      &=& (1/3- x_1/2)3/2 + x_1 9/7 + (2/3 -x_1/2) 11/9 \\
      &= & 71/54 +x_1 ( 9/7 -3/4 -11/18) \\
      &\leq & 71/54  = 4/3- 1/54.
\end{eqnarray*}

Now consider the case in which there is an 8-cycle $C_j$ of an initial cover $\C_j$ containing $V(\gamma)$. Then $v$ belongs to the structure depicted
in Figure~\ref{fig:ciclo4prob}, where $e_1 \neq e_3$, $e_2 \neq e_4$ and $e_1, e_2, e_3, e_4$ are in some matching $M_j$. As we assumed that
$|V(G)|\geq 10$, we cannot simultaneously have $e_1=e_4$ and $e_2=e_3$.
Let $f$ and $g$ be the leftmost and rightmost edge in the figure. Let also $Y = \{i\colon f \in M_i\}$ and
$Z=\{i\colon g \in M_i\}$. It is easy to check that $Y \cup Z \subseteq X_2$.

Consider an index $i \in X_2$. If $i \in Y\cup Z$ (i.e., if at least one of $f$ and $g$ are in $M_i$), then the cycle containing $v$ in $\C_i^{\UU2}$
has at least 10 vertices, and so $z_i(v)\leq 12/10=6/5$. If $i \in X_2 \setminus (Y \cup Z)$, then the cycle containing $v$ in $\C_i$ is either the
8-cycle $C_j$ of the structure, or the 8-cycle with edge set $E(C_j) \Delta E(\gamma)$. In any case $z_i(v)\leq 10/8=5/4$.

Let $y_1=\lambda(Y\cup Z)$ and $y_2=\lambda(X_2 \setminus (Y\cup Z)$, so that $y_1+y_2 = x_2$. Noting that $y_1 \geq \lambda(Y)=1/3$, we have
\begin{eqnarray*}
 z(v) &\leq & x_0 3/2 + x_1 9/7 + y_1 6/5 + (x_2-y_1) 5/4\\
      &=& (1/3- x_1/2)3/2 + x_1 9/7 + (2/3 -x_1/2) 5/4 + y_1 (6/5-5/4) \\
      &=& 4/3 - x_1(9/7-3/4-5/8) - y_1/20 \leq 4/3 - 1/60.
\end{eqnarray*}

\subsubsection{Proof of part (P4) of the Main Proposition}

Let $\gamma\in \Gamma(v)$ be an induced 5-cycle containing $v$. By Remark~\ref{remark:ciclos}, $\gamma$ is in some initial cycle cover $\C_i$. We can
assume that no 4-cycle shares exactly one edge with $\gamma$, as otherwise operation (U2), or operation (U1) before that, would have modified
$\gamma$.

The proof for this part is similar to that of part (P3). Define $X_p=\{i\colon |\gamma \cap M_i|=p\}$, for $p=0,1,2$, so that $X_0 \cup X_1 \cup
X_2=[k]$, and let $x_p=\lambda(X_p),$ for $p=0,1,2$.

By Equation~\eqref{eq:sum1} we have $x_0+x_1+x_2=1$. Applying Equation~\eqref{eq:convexMat} to the 5 edges incident to $\gamma$, we obtain
$5x_0+3x_1+x_2=5/3$. This implies that $x_0=1/2(1/3-x_1)$ and $x_2=1/2(5/3-x_1)$.

For every $i \in X_0$, we have $v\in V(\gamma)$ and $\gamma \in \C_i$. By Lemma~\ref{lema:tecnicoZiv}, $z_i(v)\leq 7/5$.
For $i \in X_1$, the fact that $\gamma$ does not share an edge with a 4-cycle implies that $v$ is in a cycle of $\C_i$ having length at least 8, and
therefore $z_i(v) \leq 10/8=5/4$.

For $i \in X_2$, we have two cases. If $\gamma$ is intersected by a single cycle $C$ of $\C_i$, then, by Remark~\ref{remark:caminoCiclo}, $C$ must be
of length at least 9, and so, $z_i(v) \leq 11/9$.

The second case is that $\gamma$ is intersected by two cycles of $\C_i$. One of them, say $C'$, shares exactly one edge with $\gamma$ (and so, $C'$
cannot be a 4-cycle), and the second one, $C''$, shares exactly two consecutive edges with $\gamma$ (by Remark~\ref{remark:caminoCiclo}, $C'$ cannot
be a 4-cycle either). Let $C \in \{C',C''\}$ be the cycle containing vertex~$v$. If $C$ is merged with another cycle during operations (U1) and (U2)
then, by Remark~\ref{remark:ciclos}, the resulting cycle containing $v$ in $\C_i^{\UU2}$ is of length at least 10, and so $z_i(v)\leq 12/10$. On the
other hand, if $C$ is not modified by operations (U1) and (U2) then, it must be modified by operation (U3) (this is because $C$ intersects the 5-cycle
$\gamma$, which in turns intersects two components of $\C_i^{\UU2}$ of length at least~5). Lemma~\ref{lema:tecnicoZiv} guarantees in this case that
$z_i(v)\leq 13/10$.

Summarizing, if $i \in X_2$, then $z_i(v)\leq \max\{ 12/10,11/9,13/10 \}=13/10$. Then,
\begin{eqnarray*}
 z(v) &\leq & x_0 7/5 + x_1 5/4 + x_2 13/10 \\
      &=&  1/2(1/3-x_1) \cdot 7/5 + x_1 5/4 + 1/2(5/3-x_1) \cdot 13/10 \\
      &=&  7/30+13/12 - x_1/10\\
      &\leq & 79/60  = 4/3-1/60.
\end{eqnarray*}

\subsubsection{Proof of part (P5)  of the Main Proposition}
Let $\gamma\in \Gamma(v)$ be an induced 6-cycle containing $v$. By Remark~\ref{remark:ciclos}, $\gamma$ is in some initial cycle cover $\C_i$. We can
assume that no 4-cycle shares exactly one edge with $\gamma$, as otherwise operations (U1) or (U2) would have modified $\gamma$ and so, by the end of
the algorithm $\gamma$ would not be a 6-cycle.

We can also assume that $\gamma$ does not intersect the 5-cycles contained in an initial cycle cover. Indeed, if this was not the case, define
$S_5=\{w\in V(\gamma)\colon$ $w$ is in some 5-cycle $C$ of an initial cycle cover$\}$.
If $w \notin S_5$ then in every initial cover, the cycle containing $w$ is of length at least 6; usingLemma~\ref{lema:tecnicoZiv}, part (P4) of the
Main Proposition, and the fact that $S_5 \neq \emptyset$ implies $|S_5|\geq 2$, we conclude that
$\sum_{w \in V(\gamma)}z(w) \leq   |S_5|\left(\frac{4}{3}-\frac{1}{60}\right) +  |V(C)\setminus S_5|\frac{4}{3} \leq
6\left(\frac{4}{3}-\frac{1}{180}\right)$,
and also that $z(w) \leq 4/3$ for all $w \in V$.

Under the assumptions above, all the components containing $v$ in the final family of covers have length at least 6. Using Lemma~\ref{lema:tecnicoZiv}
we conclude not only that $z(v) \leq \max\{13/10, 8/6\} = 4/3$ (which proves the first statement of P5) but also that $z(w) \leq 4/3$ for the 6
vertices $w \in V(\gamma)$.

Let us continue with the proof. Denote the edges of $\gamma$ as $a_1, \dots, a_6$ and the 6 edges incident to $\gamma$ as $e_1, \dots, e_6$, as in
Figure~\ref{fig:hex_C}.\label{Change Figure}
\begin{figure}[t]
\centering
\includegraphics[width=.25\textwidth]{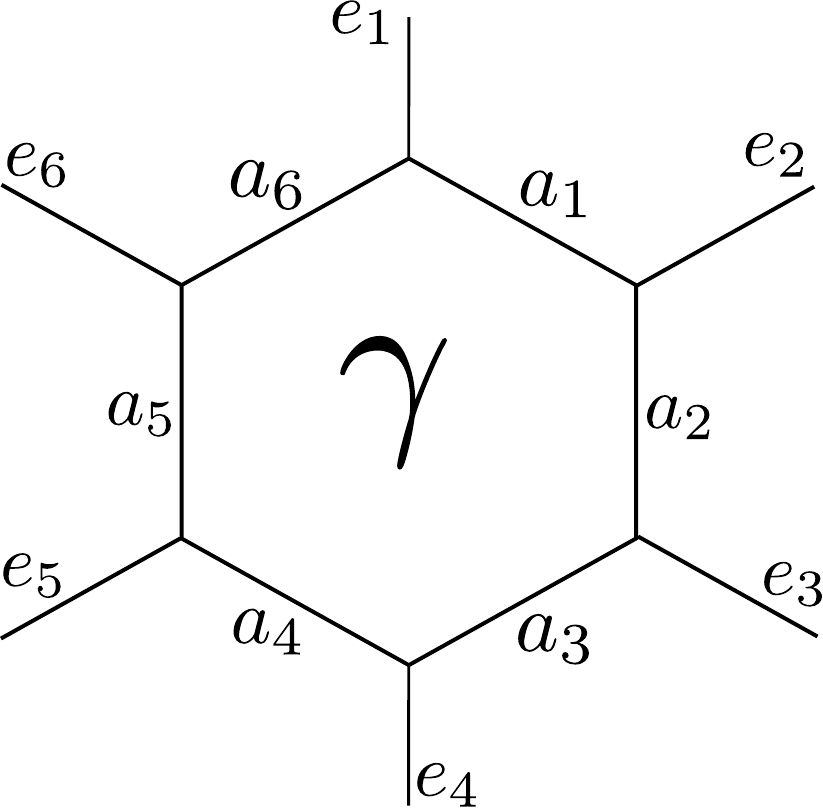}
\caption{Induced 6-cycle $\gamma$.}
\label{fig:hex_C}
\end{figure}

We now define some sets of indices according on how $\gamma$ intersects the matchings $M_1, \dots, M_k$.
For every symbol $Z \in \{X_0\} \cup \{X_1^q\}_{q=1}^6 \cup \{X_2^q\}_{q=1}^3 \cup \{Y_2^q\}_{q=1}^6 \cup \{X_3^q\}_{q=1}^2$, we define $Z$ as the set
of indices $i$ for which the matching $M_i$ contains the bold edges indicated in Figure~\ref{fig:hex_casos}.
For example, $X_0 =\{i \colon \{e_1,\ldots, e_6\} \in M_i\}$, $X_3^1 = \{i \colon \{a_1,a_3,a_5\}\in M_i\}$, and so on. Let also $x_0 = \lambda(X_0)$,
$x_i^q= \lambda(X_i^q)$ and $y_2^q = \lambda(Y_i^q)$ for every $i$ and $q$ and define
\begin{eqnarray*}
 x_1=\sum_{q=1}^6 x_1^q,~~~~ x_2=\sum_{q=1}^3 x_2^q,~~~~ y_2=\sum_{q=1}^6 y_2^q,~~~~ x_3=\sum_{q=1}^2 x_3^q,~~~~ \overline{x}_2=x_2+y_2.
\end{eqnarray*}

\begin{figure}[t]
\centering
\includegraphics[width=.1\textwidth]{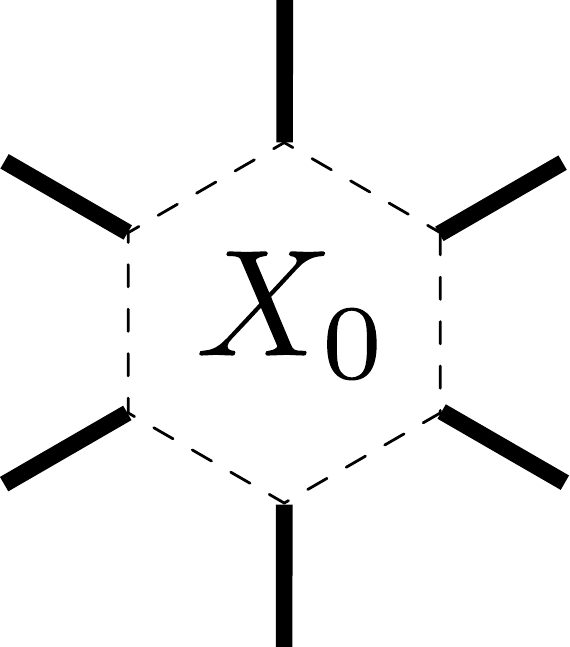}
\hspace{.3in}
\includegraphics[width=.1\textwidth]{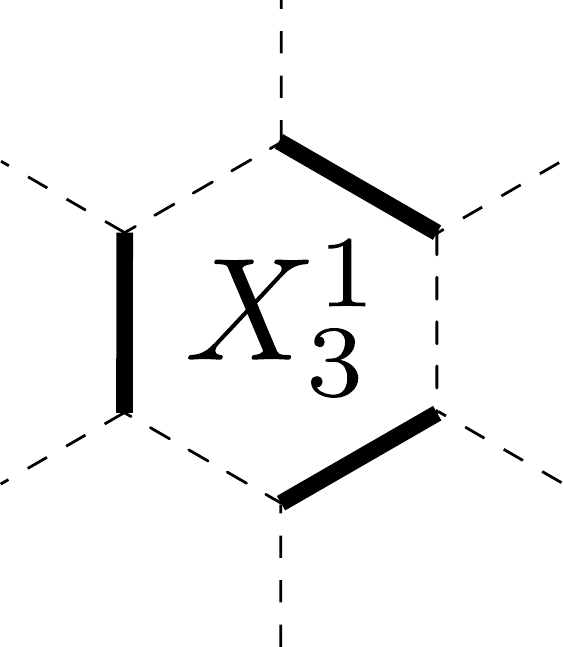}
\hspace{.3in}
\includegraphics[width=.1\textwidth]{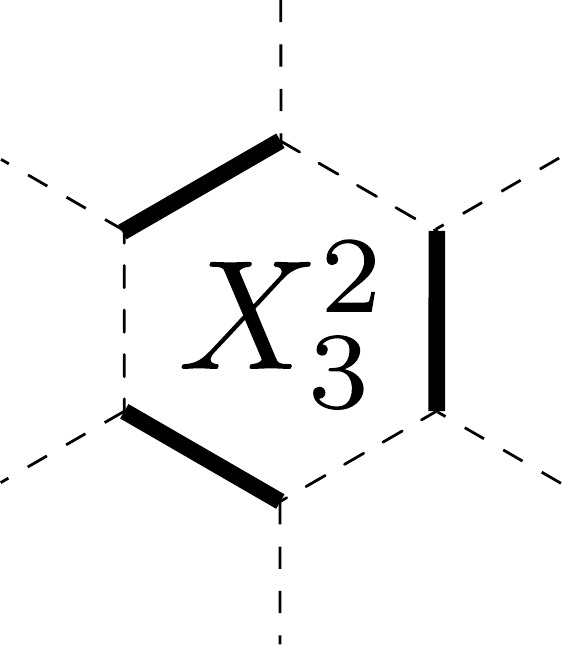}
\hspace{.3in}
\includegraphics[width=.1\textwidth]{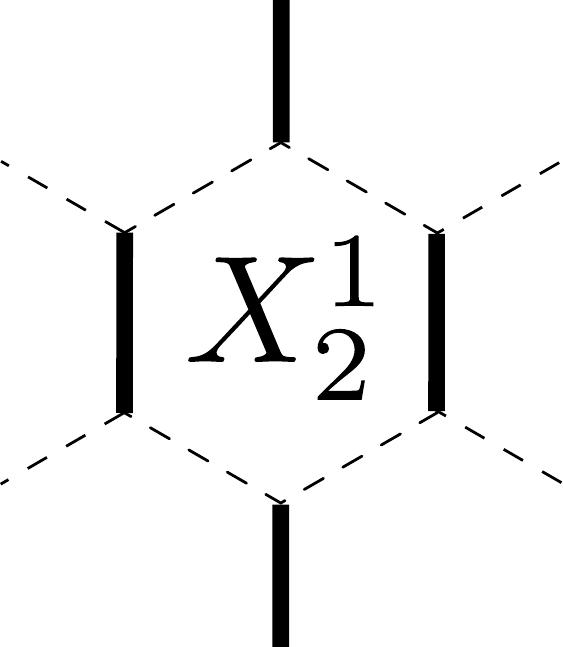}
\hspace{.3in}
\includegraphics[width=.1\textwidth]{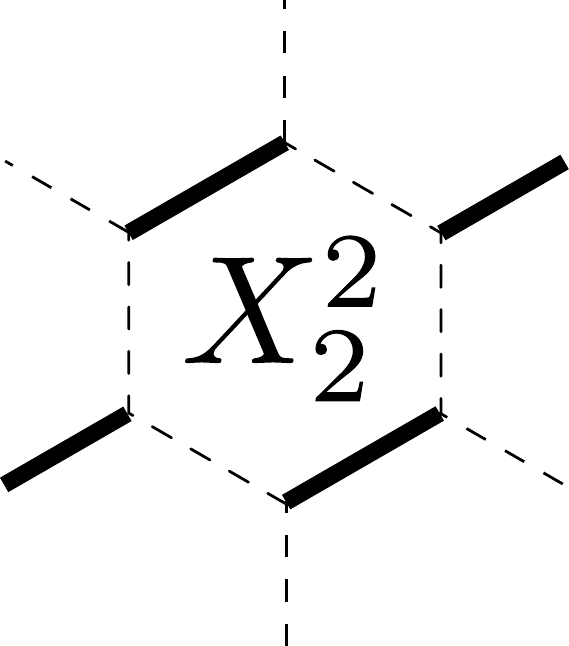}
\hspace{.3in}
\includegraphics[width=.1\textwidth]{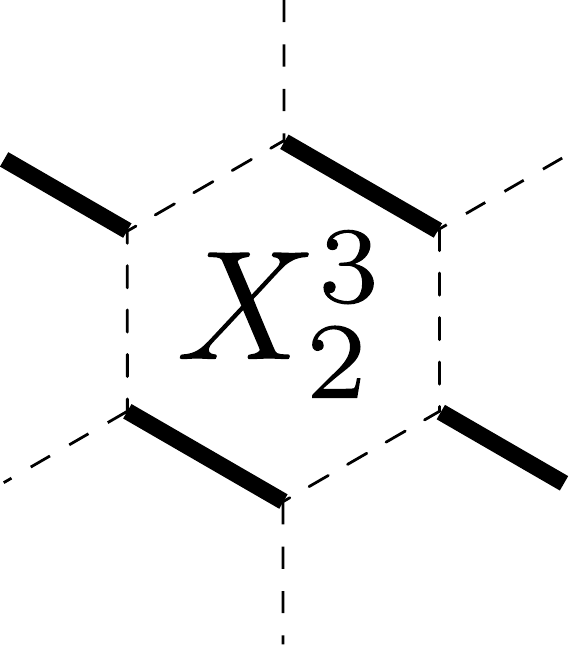}  \\
\vspace{3pt}
\includegraphics[width=.1\textwidth]{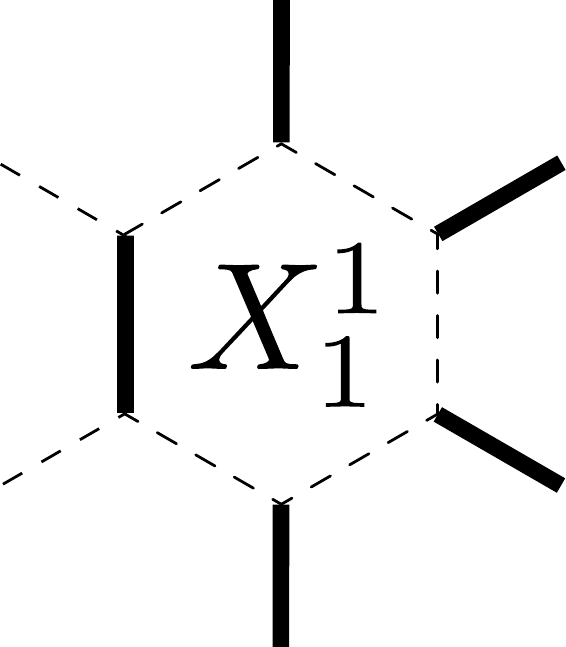}
\hspace{.3in}
\includegraphics[width=.1\textwidth]{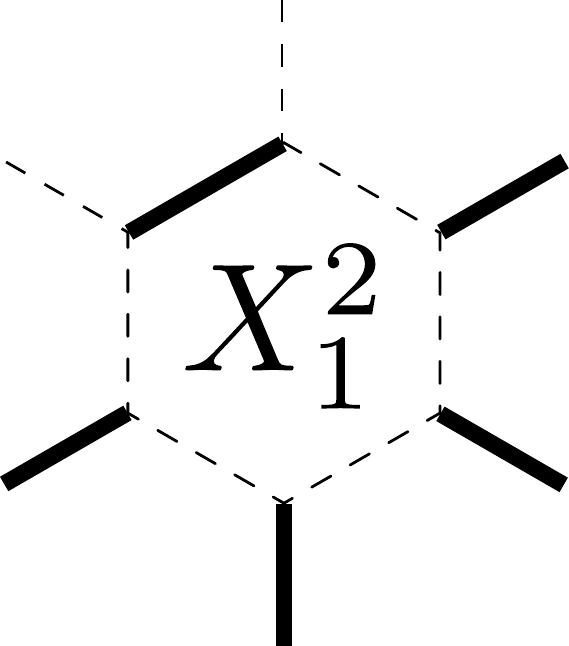}
\hspace{.3in}
\includegraphics[width=.1\textwidth]{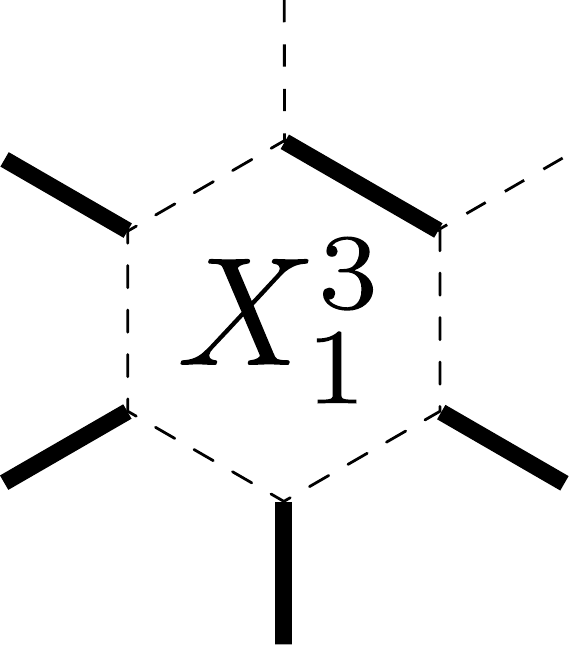}
\hspace{.3in}
\includegraphics[width=.1\textwidth]{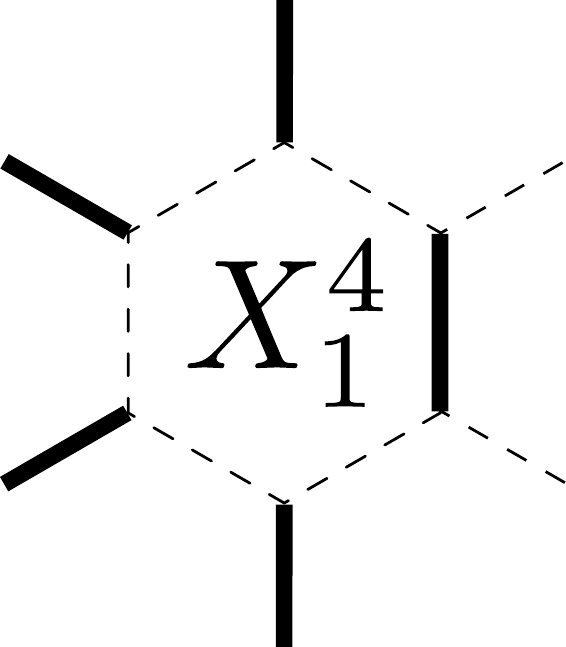}
\hspace{.3in}
\includegraphics[width=.1\textwidth]{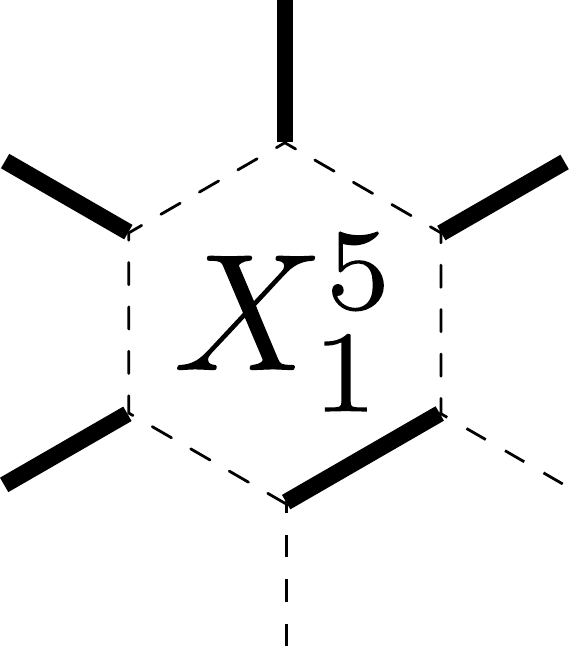}
\hspace{.3in}
\includegraphics[width=.1\textwidth]{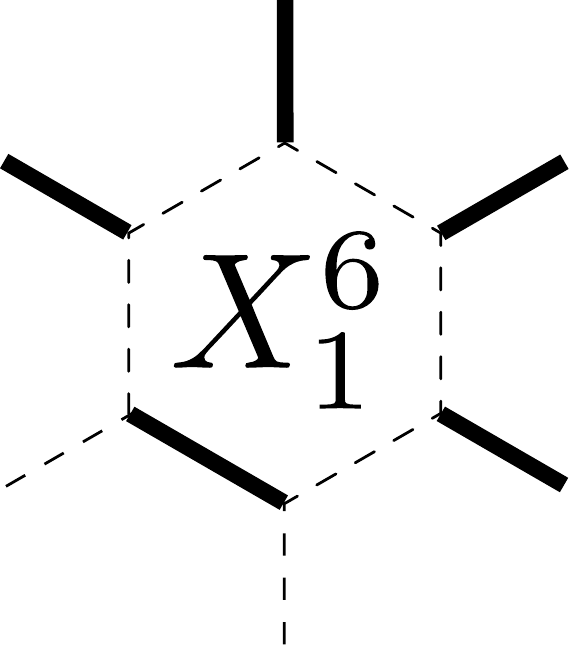}  \\
\vspace{3pt}
\includegraphics[width=.1\textwidth]{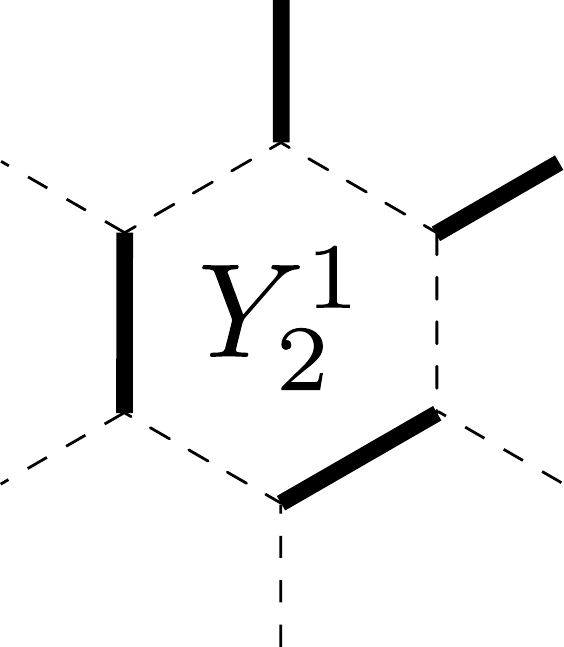}
\hspace{.3in}
\includegraphics[width=.1\textwidth]{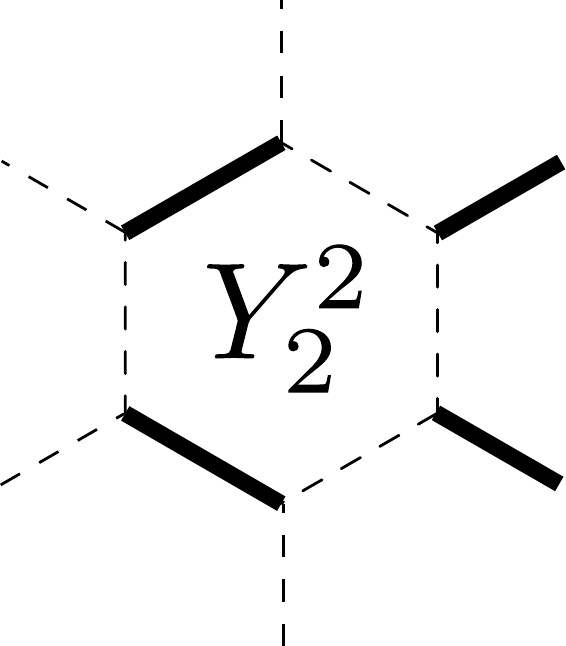}
\hspace{.3in}
\includegraphics[width=.1\textwidth]{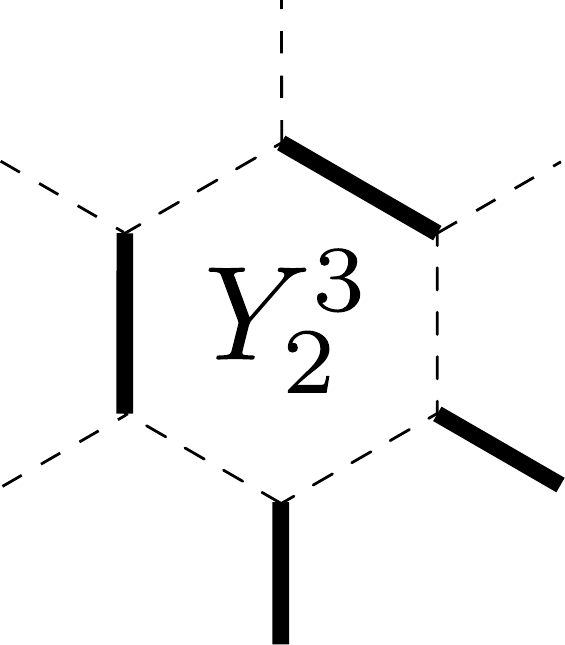}
\hspace{.3in}
\includegraphics[width=.1\textwidth]{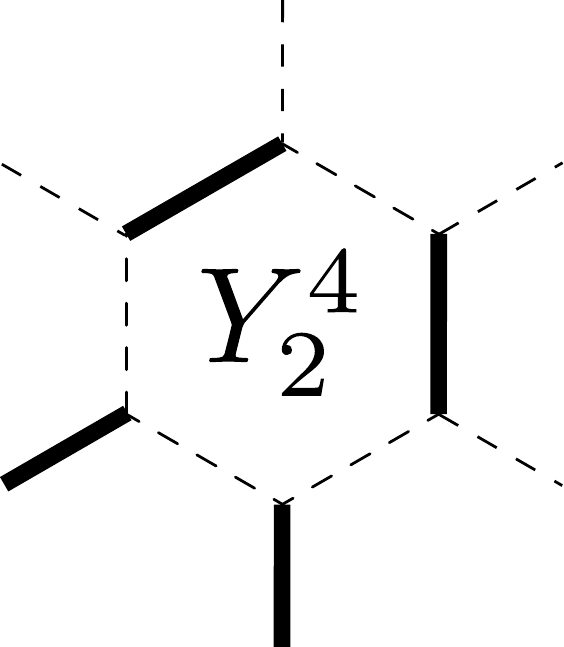}
\hspace{.3in}
\includegraphics[width=.1\textwidth]{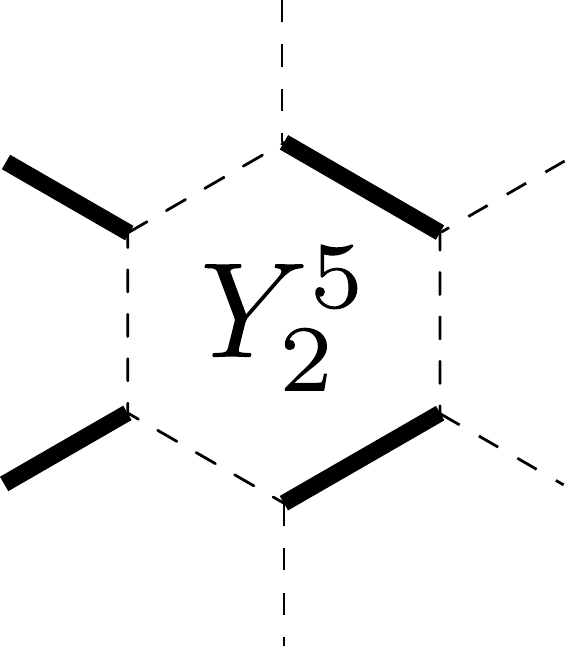}
\hspace{.3in}
\includegraphics[width=.1\textwidth]{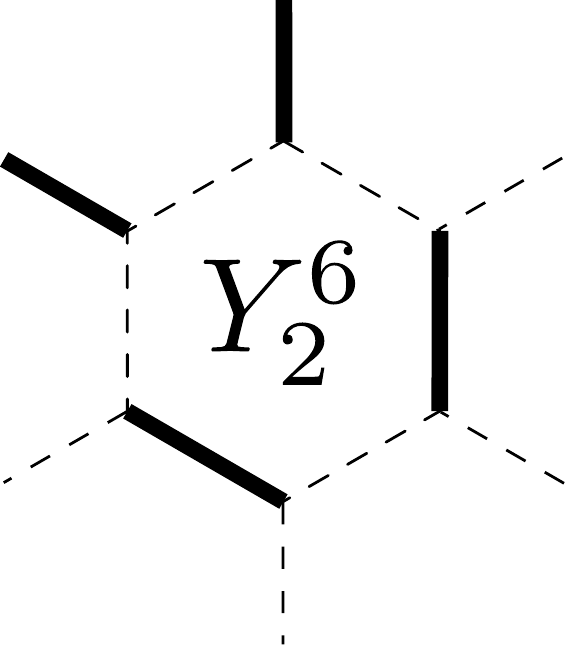} \\
\vspace{-5pt}
\caption{The different ways in which a matching can intersect $\gamma$. We use the same orientation as that of Figure~\ref{fig:hex_C}.}
\label{fig:hex_casos}
\end{figure}

Equation \eqref{eq:sum1} implies that $x_0+x_1+\overline{x}_2+x_3=1$. Equation~\eqref{eq:convexMat} applied to the set $\{a_1,\ldots, a_6\}$ of edges
incident to $\gamma$ implies that $6x_0+4x_1+2\overline{x}_2=6/3$. Hence, $3x_0+2x_1+\overline{x}_2=1$. It follows that
\begin{eqnarray}
2x_0+x_1=x_3. \label{eq:descarte}
\end{eqnarray}
Recall that there are no 4-cycles in $G$ and no 5-cycles in an initial cycle cover interseting $\gamma$ in exactly one edge. Consider $w \in
V(\gamma)$ and $i \in [k]$.

If $i \in X_0$ (i.e., $M_i$ shares no edge with $\gamma$) then $w \in V(\gamma)$ and $\gamma \in \C_i$. By Lemma~\ref{lema:tecnicoZiv} we have,
$z_i(w)\leq 8/6$. If $i \in X_1:=\cup_{q=1}^6 X_1^q$ (i.e., $M_i$ contains exactly one edge of $\gamma$) then, as no 4-cycle shares exactly one edge
with $\gamma$, $w$ must be in a cycle $C \in \C_i$ of length at least 9; therefore, $z_i(w)\leq 11/9$. If $i \in X_3:=\cup_{q=1}^2 X_3^q$ (i.e., $M_i$
contains three edges of $\gamma$) then we have two cases. The first case is that $\gamma$ is intersected by 1 or 3 cycles of $\C_i$. Then, by the end
of operation (U1), $w$ must be in a cycle of $\C_i^{\UU1}$ of length at least 9 and so
$z_i(w)\leq 11/9$. The second case is that $\gamma$ is intersected by 2 cycles of $\C_i$. One of them shares exactly 2 edges with $\gamma$, thence it
must be of length at least 8. The other cycle shares exactly one edge with $\gamma$ and so it must be of length at least 6. Therefore, in this case, 4
of the vertices $w$ of $\gamma$ satisfy $z_i(w) \leq 10/8$ and the remaining 2 satisfy $z_i(w) \leq 8/6$.

We still need to analyze the indices $i \in X_2:=\cup_{q=1}^3 X_2^q$ and $i \in Y_2:=\cup_{q=1}^6 Y_2^q$ (i.e., those for which $M_i$ shares two edges
with $\gamma$). Let $0<\delta \leq 1$ be a constant to be determined. We divide the rest of the proof in two scenarios.

\noindent\textbf{Scenario 1:} If $x_3$ (which equals $\max\{x_0,x_1,x_3\}$ by \eqref{eq:descarte}) is at least $\delta$.

If $i \in X_2 \cup Y_2$, then every vertex $w \in \gamma$ is in a cycle $C \in \C_i$ of length at least 6; therefore $z_i(w)\leq 8/6$ and
{\small\begin{eqnarray}
\sum_{w \in V(\gamma)}z(w) &\leq & 6\cdot(x_0 8/6 + x_1 11/9+ \overline{x}_2 8/6) +   x_3\left( 2\cdot \frac{8}{6} +  4\cdot \frac{10}{8} \right)
\notag \\
                      &\leq &  6\cdot  \left((1-x_3) 4/3   +x_3\left( \frac{4}{3}-\frac{1}{18}\right)\right) \leq   6\cdot
\left(4/3   -\delta /{18}\right). \label{ineq:beta1}
\end{eqnarray}}
\noindent\textbf{Scenario 2:} If $x_3$ (which equals $\max\{x_0,x_1,x_3\}$ by \eqref{eq:descarte}) is at most $\delta$.

We start by stating the following technical lemma.

\begin{lem}
\label{lemma:betaTec}
Define $\beta:=1/9-\delta$. Then at least one of the following cases hold:

\begin{tabular}{ll}
\begin{minipage}{.4\textwidth}
 \begin{itemize}
\item[-]{\bf Case 1:\ } $x_2^1, x_2^2, x_2^3 \geq \beta$.
\item[-]{\bf Case 2:\ } $x_2^1, y_2^2, y_2^5 \geq \beta$.
\item[-]{\bf Case 3:\ } $x_2^2, y_2^3, y_2^6 \geq \beta$.
\item[-]{\bf Case 4:\ } $x_2^3, y_2^1, y_2^4 \geq \beta$.
\end{itemize}
\end{minipage}
&
\begin{minipage}{.45\textwidth}
\begin{itemize}
\item[-]{\bf Case 5:\ } $y_2^1, y_2^4, y_2^2, y_2^5 \geq \beta$.
\item[-]{\bf Case 6:\ } $y_2^2, y_2^5, y_2^3, y_2^6 \geq \beta$.
\item[-]{\bf Case 7:\ } $y_2^1, y_2^4, y_2^3, y_2^6 \geq \beta$.
\end{itemize}
\end{minipage}
\end{tabular}
\end{lem}
\medskip
\begin{proof}
By applying \eqref{eq:convexMat} on edges $e_1$ and $a_2$ respectively (see Figure~\ref{fig:hex_C}) we get
\begin{eqnarray}
x_0+ x_1^1 +x_1^4 +x_1^5 +x_1^6+x_2^1+y_2^1+y_2^6&=&\frac{1}{3}. \label{eq:tec1}\\
x_1^4+x_2^1+y_2^4+y_2^6+x_3^2  &=&\frac{1}{3}. \label{eq:tec2}
\end{eqnarray}
Substracting \eqref{eq:tec1} and \eqref{eq:tec2}, using $\max\{x_0,x_1,x_3\}\leq \delta$ and equation~\eqref{eq:descarte} we obtain
\begin{eqnarray}
|y_2^1-y_2^4| &\leq& \delta    \label{eq:tec3}.
\end{eqnarray}

Analogously, we also have
\begin{eqnarray}
|y_2^2-y_2^5| &\leq& \delta,    \label{eq:tec4} \\
|y_2^3-y_2^6| &\leq& \delta.    \label{eq:tec5}
\end{eqnarray}

Using $\max\{x_0,x_1,x_3\}\leq \delta$, equation \eqref{eq:descarte} and applying \eqref{eq:convexMat} on edge $e_j$, for $j \in
\{1,...,6\}$ we have
\begin{eqnarray}
 x_2^{1}+y_2^{1}+y_2^{6} &\geq &1/3 -\delta, \label{eq:tec11} \\
 x_2^{2}+y_2^{2}+y_2^{1} &\geq &1/3 -\delta, \label{eq:tec12}\\
 x_2^{3}+y_2^{3}+y_2^{2} &\geq &1/3 -\delta, \label{eq:tec13}\\
 x_2^{1}+y_2^{4}+y_2^{3} &\geq &1/3 -\delta, \label{eq:tec14}\\
 x_2^{2}+y_2^{5}+y_2^{4} &\geq &1/3 -\delta, \label{eq:tec15}\\
 x_2^{3}+y_2^{6}+y_2^{5} &\geq &1/3 -\delta, \label{eq:tec16}
\end{eqnarray}

Now we are ready to prove the lemma. Assume by sake of contradiction that none of the cases in the lemma holds. As case 1 does not hold, we can
assume without loss of generality that one of the
following is true.
\begin{enumerate}
\item[(i)] $x_2^1<\beta$, $x_2^2, x_2^3 \geq \beta$,
\item[(ii)] $x_2^1,x_2^2<\beta$, $ x_2^3 \geq \beta$,
\item[(iii)] $x_2^1,x_2^2,x_2^3<\beta$.
\end{enumerate}

Consider the case that (i) is true. Since Case 3 does not hold and $x_2^2 \geq \beta$ we conclude that $\min\{y_2^3,y_2^6\}<\beta$. Using
Inequality~\eqref{eq:tec5} we get $y_2^3, y_2^6<\beta+\delta$. Analogously, since Case 4 does not hold and $x_2^3 \geq \beta$ we conclude that
$\min\{y_2^1,y_2^4\}<\beta$. Using Inequality~\eqref{eq:tec3} we get $y_2^1, y_2^4<\beta+\delta$. Then we have
$$x_2^1+y_2^1+y_2^6<3\beta+2\delta=1/3-\delta,$$
which contradicts inequality~\eqref{eq:tec11}.

Consider the case that (ii) is true. Similar as above, since $x_2^3\geq \beta$ and Case 4 does not hold we conclude that $y_2^1,y_2^4<\beta+\delta$.
Furthermore, using inequality~\eqref{eq:tec4} and that Case 6 does not hold, we have at least one of the following inequalities
$y_2^2,y_2^5<\beta+\delta$ or $y_2^3,y_2^6 < \beta + \delta$. If the first one is true then,
$$x_2^2+y_2^2+y_2^1<3\beta+2\delta=1/3-\delta,$$
which contradicts Inequality~\eqref{eq:tec12}. If the second one is true, then
$$x_1^2+y_2^2+y_2^6<3\beta+2\delta=1/3-\delta,$$
which contradicts Inequality~\eqref{eq:tec11}.

Finally, consider the case that (iii) is true. As Cases 5, 6 and 7 do not hold, we have that $\min\{y_2^1, y_2^4, y_2^2, y_2^5\}<\beta$,
$\min\{y_2^2, y_2^5, y_2^3, y_2^6\}<\beta$ and $\min\{y_2^1, y_2^4, y_2^3, y_2^6\}<\beta$. Without loss of generality, we can assume that $y_2^1,
y_2^2 < \beta$. Using inequalities ~\eqref{eq:tec3}~and~\eqref{eq:tec4} we conclude that $y_2^1,y_2^4<\beta+\delta$ and $y_2^2,y_2^5<\beta+\delta$.
Therefore,
$$x_2^2+y_2^2+y_2^1<3\beta+2\delta=1/3-\delta,$$
which contradicts inequality \eqref{eq:tec12}.
\end{proof}

Denote an index $i \in X_2\cup Y_2$ as \emph{long} if there are at least 2 vertices of $V(\gamma)$ contained in a single cycle of $\C_i^{\UU1}$ of
length at least 7, otherwise denote it as \emph{short}.  A set $Z \subseteq [k]$ is called long if $Z$ contains only long indices.

Consider a short index $i \in X_2 \cup Y_2$. Since the matching $M_i$ contains two edges of $\gamma$, we must be in the case where $\gamma$ intersects
exactly two cycles of $\C_i^{\UU1}$ and both of them are 6-cycles (we assumed at the beginning of the proof of this part that no cycle in $\C_i$ of
length at most 5 intersects $\gamma$). The next lemma complements what happens in each of the cases introduced in Lemma~\ref{lemma:betaTec}.

\begin{lem}
\label{lemma:betaTec2}$~$\\[-3ex]
\begin{itemize}
\item[(1)] If $X_2^1$, $X_2^2$ and $X_2^3$ are non-empty then at least one of them is long.
\item[(2)] If $X_2^1$, $Y_2^2$ and $Y_2^5$ are non-empty then at least one of them is long.
\item[(3)] If $X_2^2$, $Y_2^1$ and $Y_2^4$ are non-empty then at least one of them is long.
\item[(4)] If $X_2^3$, $Y_2^3$ and $Y_2^6$ are non-empty then at least one of them is long.
\item[(5)] If $Y_2^1$, $Y_2^4$, $Y_2^2$ and $Y_2^5$ are non-empty then at least one of them is long.
\item[(6)] If $Y_2^2$, $Y_2^5$, $Y_2^3$ and $Y_2^6$ are non-empty then at least one of them is long.
\item[(7)] If $Y_2^1$, $Y_2^4$, $Y_2^3$ and $Y_2^6$ are non-empty then at least one of them is long.
\end{itemize}
\end{lem}
\begin{proof}
We only prove items 1, 2 and 5, since the proofs for the rest are analogous.
\begin{itemize}
\item[(1)] Assume for contradiction that there are short indices $i_1 \in X_2^1$, $i_2 \in X_2^2$ and $i_3 \in X_3^3$. In particular, every vertex of
$\gamma$ is in a 6-cycle of $\C_{i_p}^{\UU1}$ (and thus, of $\C_{i_p}$) for $p=1,2,3$. From this, we deduce that the neighborhood of $\gamma$ in $G$
is exactly as depicted in Figure~\ref{fig:hex_interX2}.
\begin{figure}[t]
\begin{minipage}[b]{.48\textwidth}
\centering
\includegraphics[width=.45\textwidth]{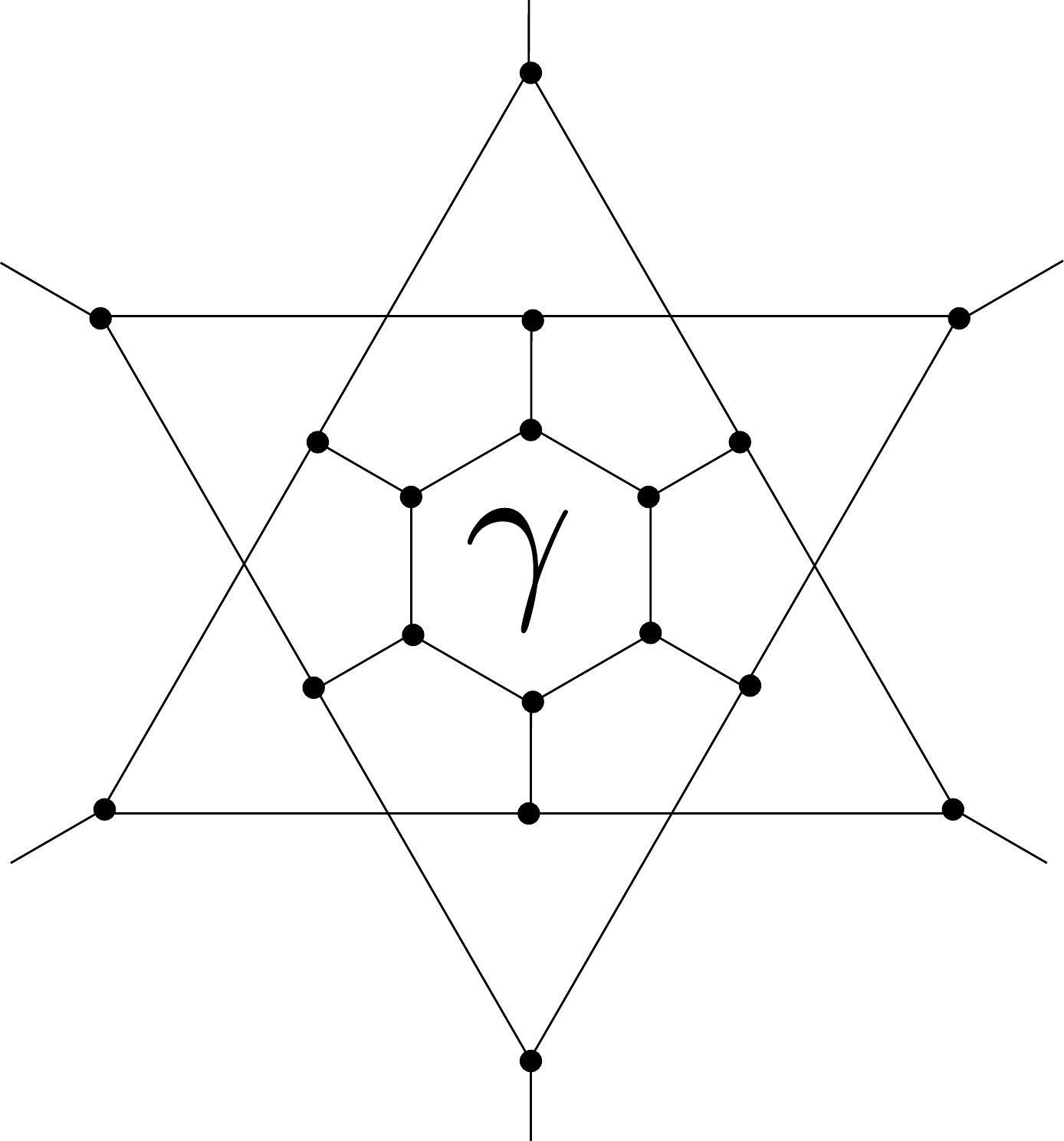}
\caption{6-cycle $\gamma$ for the case in which $X_2^1$, $X_2^2$ and $X_2^3$ are nonempty and not long.}
\label{fig:hex_interX2}
\end{minipage} \hfill
\begin{minipage}[b]{.48\textwidth}
\centering
\includegraphics[width=.9\textwidth]{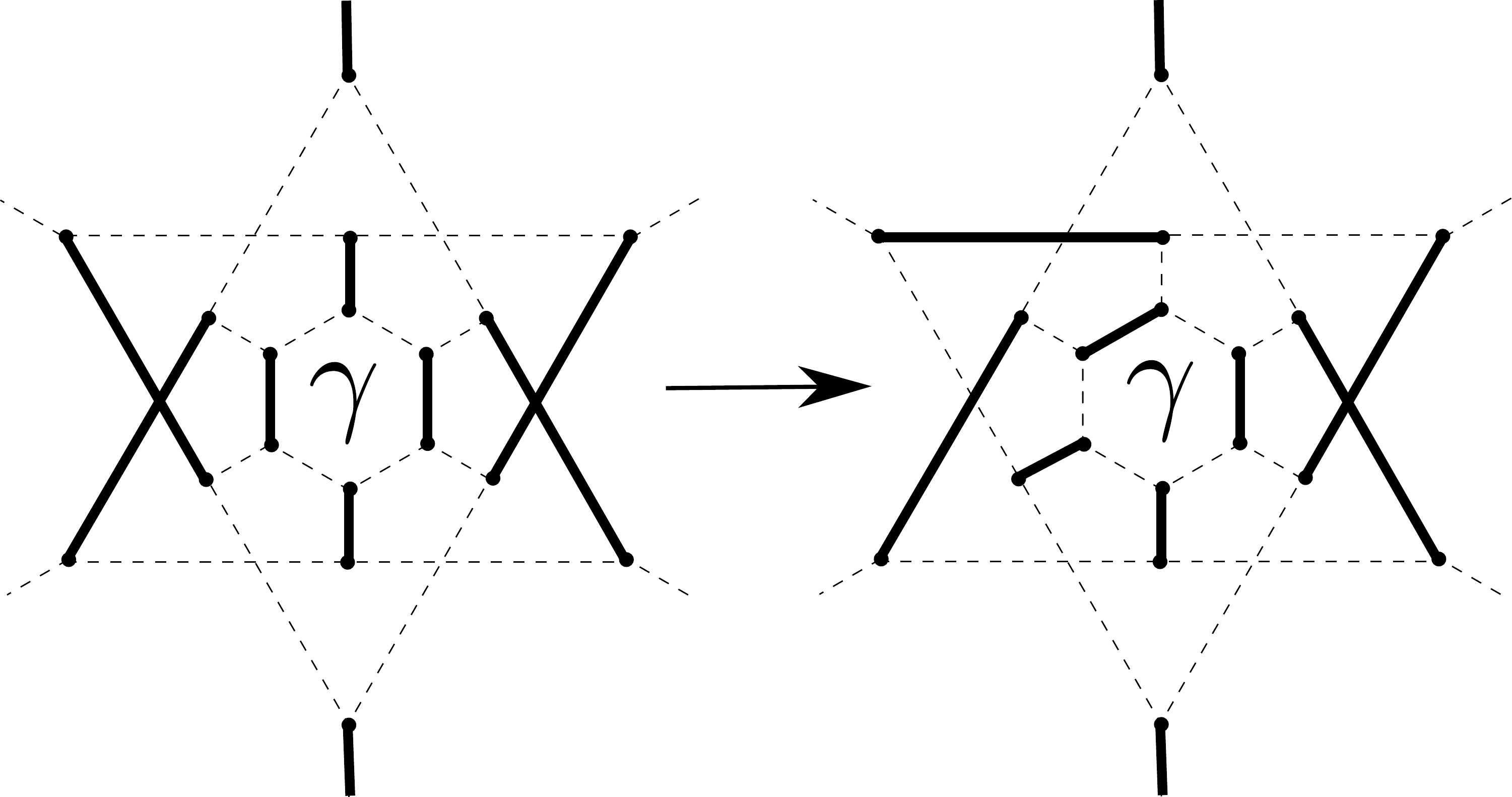}
\caption{Operation (U1) applied to cycles in $\C_{i_1}$, where $i_1$ is a short index of $X_2^1$.}
\label{fig:hex_interX2opi}
\end{minipage}
\end{figure}
 Now focus on the short index $i_1 \in X_2^1$. Since $G$ is as in Figure~\ref{fig:hex_interX2}, there are three cycles of $\C_{i_1}$ sharing each one
edge with a 6-cycle of $G$. But then, as Figure~\ref{fig:hex_interX2opi} shows, operation (U1) would have merge them into a unique
cycle $C$ in $\C_{i_1}^{\UU1}$ of length at least 16, contradicting the fact that $i_1$ is short.
\item[(2)] Assume for contradiction that there are short cycles $i_1 \in X_2^i$ $i_2 \in Y_2^2$ and $i_3 \in Y_2^5$. In particular, every vertex of
$\gamma$ is in a 6-cycle of $\C_{i_p}^{\UU1}$ (and thus, of $\C_{i_p}$) for $p=1,2,3$.  From this, we deduce that the neighborhood of $\gamma$ in $G$
is exactly as depicted in Figure~\ref{fig:hex_interXY2},
\begin{figure}[t]
\begin{minipage}[b]{.31\textwidth}
\centering
\includegraphics[width=.6\textwidth]{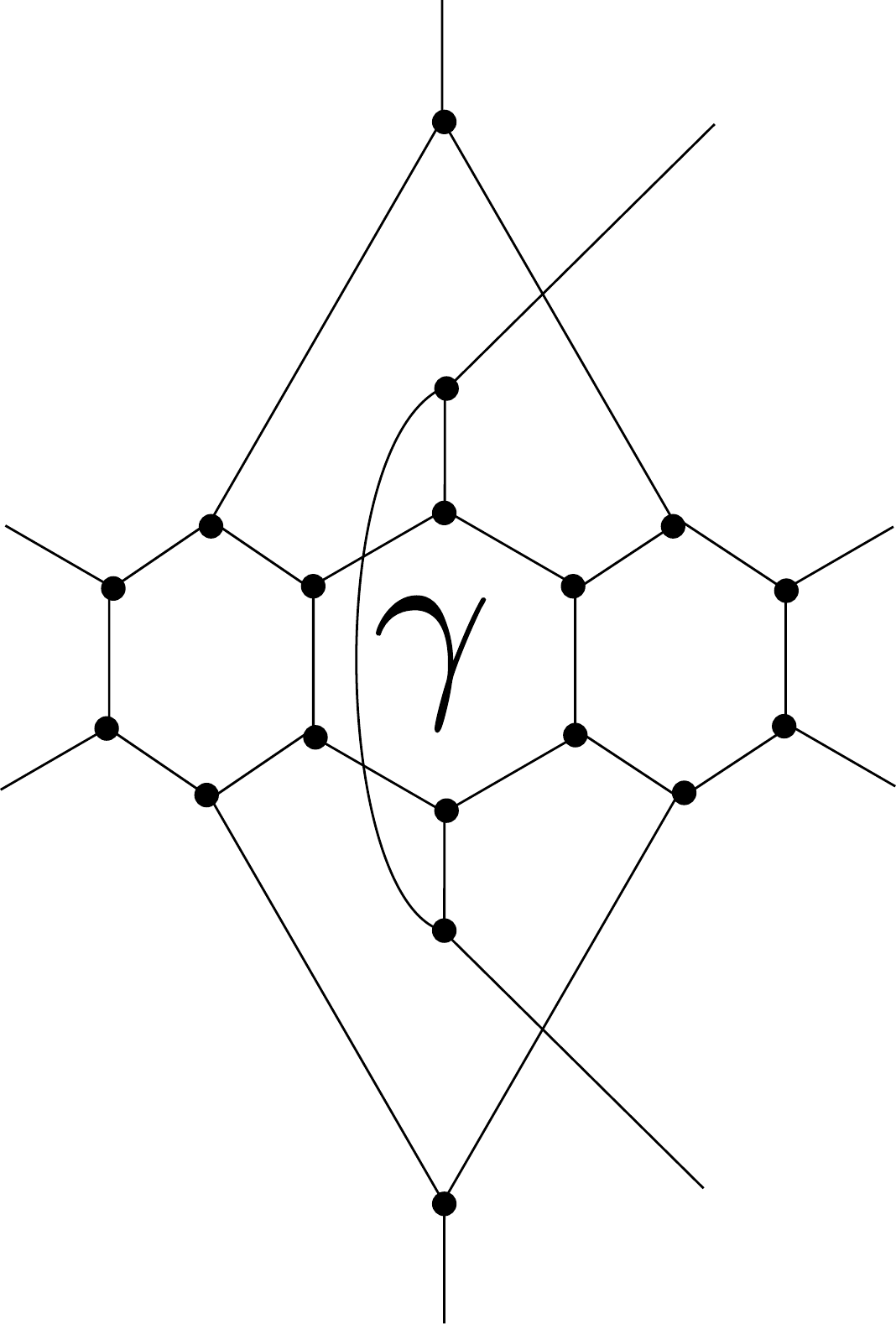}
\caption{6-cycle $\gamma$ for the case $X_2^1$, $Y_2^2$, $Y_2^5$ are nonempty and not long.}
\label{fig:hex_interXY2}
\end{minipage} \hfill
\begin{minipage}[b]{.35\textwidth}
\centering
\includegraphics[width=\textwidth]{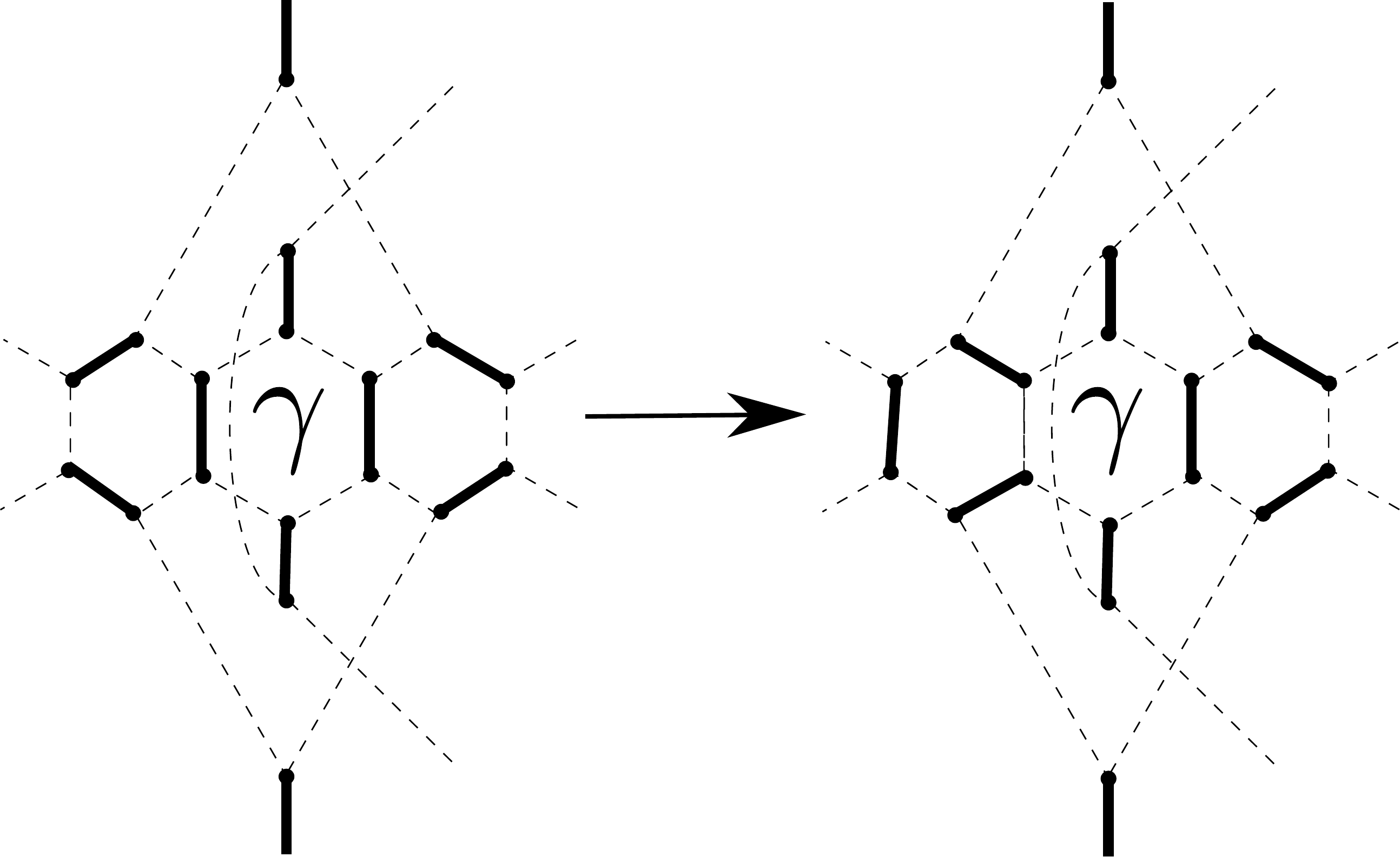}
\caption{Operation (U1) applied to cycles in $\C_{i_1}$, where $i_1$ a short index of $X_2^1$.}
\label{fig:hex_interXY2opi}
\end{minipage} \hfill
\begin{minipage}[b]{.31\textwidth}
\centering
\includegraphics[width=.8\textwidth]{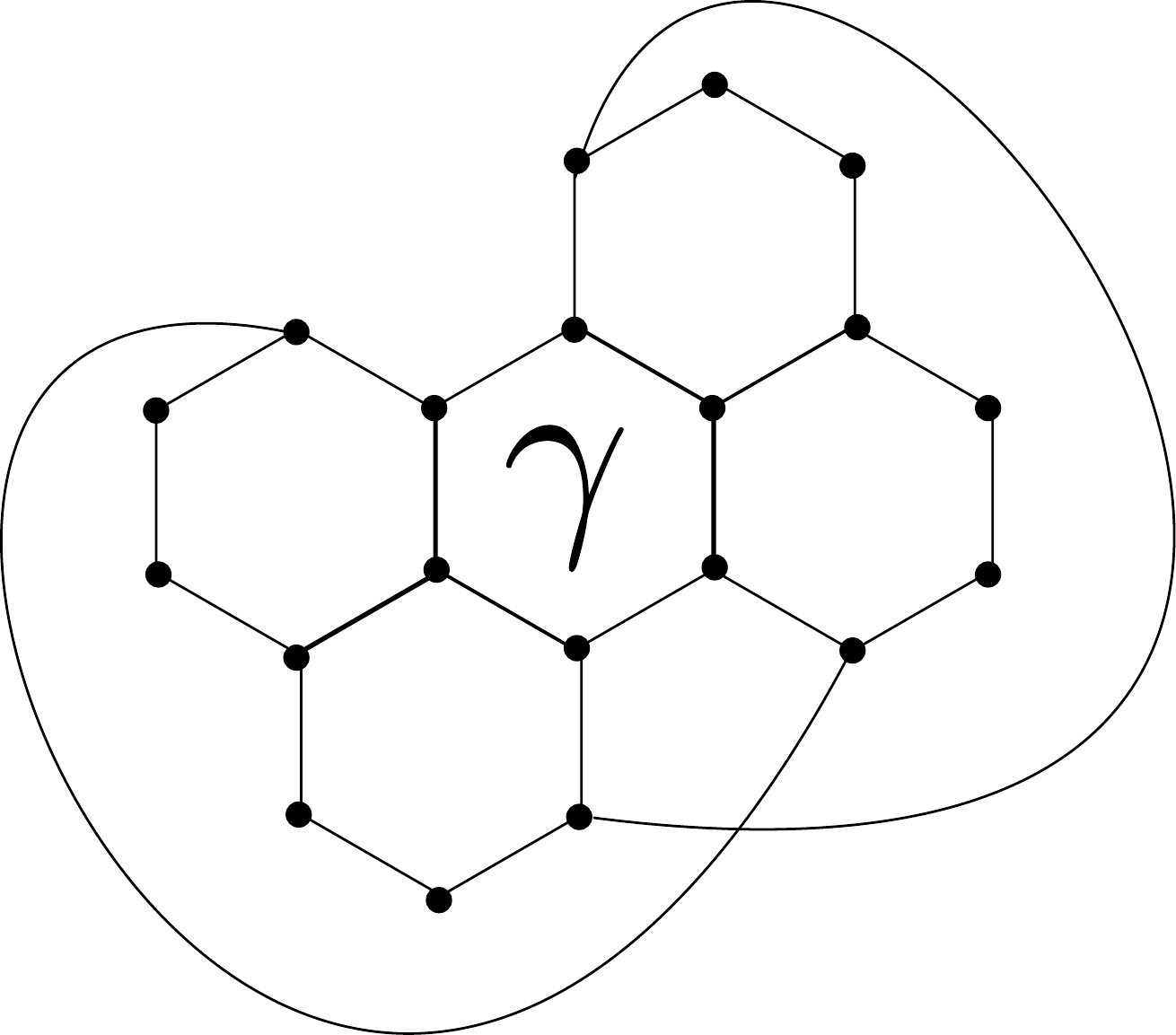}
\caption{6-cycle $\gamma$ for the case  $Y_2^1$, $Y_2^4$, $Y_2^2$, $Y_2^5$ are non-empty and not long.}
\label{fig:hex_interY2}
\end{minipage}
\end{figure}

Focus on the short index $i_1 \in X_2^1$. Since $G$ is as in Figure~\ref{fig:hex_interXY2}, there are three cycles of $\C_{i_1}$ that share one edge
each with a 6-cycle of $G$. But in this case, as Figure~\ref{fig:hex_interXY2opi} shows, operation (U1) would have merge them into a unique cycle $C$
in $\C_{i_1}^{\UU1}$ of length at least 16, contradicting the fact that $i_1$ is short.

\item[(5)] Assume for contradiction that there are short indices $i_1 \in Y_2^1$, $i_2\in Y_2^4$, $i_3\in Y_2^2$ and $i_4 \in Y_2^5$. In particular
every vertex of $\gamma$ is in a 6-cycle of $\C_{i_p}^{\UU1}$ (and thus, of $\C_{i_p}$), for $p=1,2,3,4$. Then, the neighborhood of $\gamma$ in $G$ is
exactly as depicted in Figure~\ref{fig:hex_interXY2}. But this structure shows a contradiction, as matching $M_{i_1}$ can not be completed to the
entire graph. \qedhere
\end{itemize}
\end{proof}

Using Lemmas~\ref{lemma:betaTec} and~\ref{lemma:betaTec2} we conclude that there is a long set of indices $Z \subseteq X_2 \cup Y_2$ for which
$\lambda(Z)\geq \beta$. In particular, using Lemma~\ref{lema:tecnicoZiv}, we conclude that for every $i \in Z$, there are 2 vertices $w$ in $\gamma$
with $z_i(w) \leq 9/7$, and for the remaining four vertices of $\gamma$, $z_i(w)\leq 4/3$.
Altogether, $\sum_{w \in V(\gamma)}z(w)$ is at most
{\small
\begin{align}
&6\cdot\!\biggl(x_0 \frac86  + x_1 \frac{11}{9}+(\overline{x}_2-\beta)\frac86\biggr) +  \beta \left( 2\!\cdot\!\frac{9}{7} +  4\!\cdot\! \frac{8}{6}
\right) +   x_3\left( 2\!\cdot\! \frac{8}{6} +  4\!\cdot\! \frac{10}{8} \right) \notag\\
&\leq 6(1- \beta)\frac43 + \beta \left( 2\cdot \frac{9}{7} +  4\cdot \frac{8}{6} \right) =  6\cdot\left(\frac43 -\frac{1/9-\delta}{63}\right).
\label{ineq:beta2}
\end{align}}

To end the proof, we set $\delta=2/81$, so that $(1/9-\delta)/63=\delta/18=1/729$. From Inequalities~\eqref{ineq:beta1} and \eqref{ineq:beta2} we
conclude that in any scenario,
\begin{eqnarray}
\sum_{w \in V(\gamma)}z(w) &\leq & 6\cdot(4/3 -1/729).
\end{eqnarray}

\subsubsection{Proof of part (P6) of the Main Proposition}
If none of the cases indicated by the previous parts hold then there are no 4, 5 and 6-cycles in $\Gamma(v)$. In other words, all components
containing $v$ in the final family of covers have length at least 7. By Lemma~\ref{lema:tecnicoZiv} we conclude that $z(v) \leq \max\{13/10, 9/7\}
= 13/10$.

\section{General connected cubic graphs}
In this section we give a $4/3 - \epsilon'$ approximation algorithm for the TSP of a connected cubic graph $G$, where $\epsilon'=1/183711$.
Additionally, our algorithm shows that the integrality gap of the subtour LP for general cubic graphs is also at most $4/3 - \epsilon'$.

Observe that if $G$ is not 2-connected, the number $n$ of vertices is no longer the optimum value of the subtour LP. In order to get the desired
approximation we need to consider separatedly the bridges of $G$, since every feasible tour uses them at twice.

Let $F$ be the set of bridges of $G$, and let $b=|F|$. Since $G$ is connected, the graph $G\setminus F$ is formed by exactly $(b+1)$
subcubic, 2-edge-connected components. Let $\C=\{G_1,\ldots,G_{b+1}\}$ be the collection of components of $G\setminus F$. Let 
$\C_0 \subseteq \C$
be the collection of singleton-components: they are the ones corresponding to cut-vertices of $G$ (if $G$ is 2-connected, then $|\C_0|=0$). Let also $n_i$ be the number of vertices in $G_i$ and $n_0=|\C_0|$ be the total number of singleton components.

Let $\SUB$ be the optimal subtour LP value of $G$ and $\SUB_i$ be the optimal subtour LP value of component $G_i$. Clearly, if $e$ is a bridge, the corresponding subtour LP variable has to be set to 2 in an optimal solution. Also, for every $G_i$, $\SUB_i \geq n_i$ and if $G_i$ is a singleton, then $\SUB(i)=0$. Therefore,
\begin{equation}\label{eqn:lowerbound}
\OPT \geq \SUB \geq 2b + \sum_{i=1}^{b+1} \SUB(i) \geq 2b + n-n_0,
\end{equation}
where $\OPT$ is the optimal tour value. 

The idea of our algorithm is to find a short tour in each $G_i$ and then glue the solutions into a single tour by doubling the bridges. 
Since each nonsingleton component is bridgeless and has only vertices of degree 2 and 3, we can apply M\"omke and
Svensson's algorithm~\cite{MS11} to get tour of length at most $(4/3)n_i$ on each of them. Unfortunately, that is not enough to get an overall
$(4/3-\epsilon')$-approximation for $G$. Instead, on each nonsingleton component we apply algorithms $A$ and $B$ below and
choose the solution using the fewer edges. Afterwards, we output the union of all the returned solution together with the doubled bridges.

\begin{itemize}
 \item[$A$:] Return the tour given by M\"omke and Svensson's algorithm on the component.

 \item[$B$:] Replace each vertex of degree 2 by a chorded 4-cycle, so that the resulting graph is cubic. Apply the $(4/3 - \epsilon)$-algorithm of Theorem \ref{thm:main} to the
expanded cubic 2-connected component to get a tour. Output the tour obtained by contracting the chorded 4-cycles.
\end{itemize}

\begin{thm}\label{thm:cubic1connected}
The previous algorithm returns a tour of length at most $(4/3-\epsilon')\SUB$, where $\epsilon'=\epsilon/(3+3\epsilon)=1/183711$.
\end{thm}
\begin{proof} Let $A(i)$ and $B(i)$ be the length of the tour returned by algorithms $A$ and $B$ on component $G_i$ respectively, and let $L(A)$
(respectively $L(B)$) be the total length of the tour resulting by putting together all tours $A(i)$ (respectively, $B(i)$) and twice the collection of bridges.

Using that  $A(i)\leq (4/3)n_i$, for all nonsingleton $G_i$,
\begin{align*}
L(A) &= 2b + \sum_{i=1}^{b+1} A(i) \leq 2b +  (4/3)(n-n_0).
\end{align*}

To analyze the second algorithm we need a little more of work. 
Let $D(i)$ be the number of  vertices of degree 2 in component $G_i$ before expanding it. The expanded components has $n_i + 3D(i)$ vertices. Clearly, the tour of length $B(i)$ is obtained from the tour of length $B^*(i)$ (in the expanded component) by contracting the chorded 4-cycles.
Using that $B^*(i)$ contains at least 3 edges in each chorded 4-cycle, and Theorem~\ref{thm:main}, 
\begin{align*}
B(i)&\leq B^*(i)-3D(i) \leq (\frac43 - \epsilon)(n_i + 3D(i))-3D(i) \leq (\frac43 - \epsilon)n_i + D(i).
\end{align*}
Therefore
\begin{align*}
L(B) \leq 2b + (4/3 - \epsilon)(n-n_0) + \sum_{i=1}^{b+1}D(i) \leq 4b + (4/3-\epsilon)(n-n_0).
\end{align*}
Then, the tour we return has length at most $\min(L(A),L(B))$, which can be checked to be at most 
\begin{align*}
\left(1 + \frac{1}{3(1+\epsilon)}\right)(n-n_0+2b) \leq \left(1 + \frac{1}{3(1+\epsilon)}\right)\SUB,
\end{align*}
where the last inequality follows from \eqref{eqn:lowerbound}.
\end{proof}

\small
\bibliographystyle{plain}

\end{document}